\documentclass[runningheads]{llncs}
\usepackage{cite}

\usepackage{amsmath,amsthm,amssymb,amsfonts}

\usepackage{geometry}
\DeclareMathOperator*{\argmin}{argmin}
\newcommand\ineqa{\stackrel{\mathclap{\normalfont\mbox{(a)}}}{\leq}}
\newcommand\ineqb{\stackrel{\mathclap{\normalfont\mbox{(b)}}}{\leq}}
\newcommand\ineqaa{\stackrel{\mathclap{\normalfont\mbox{(a)}}}{\geq}}
\newcommand\ineqbb{\stackrel{\mathclap{\normalfont\mbox{(b)}}}{\geq}}

\usepackage{threeparttable} 
\usepackage{mathtools}
\usepackage{algorithmic}
\usepackage{graphicx}
\usepackage{textcomp}
\usepackage{xcolor}
\usepackage{amssymb}
\usepackage{subcaption}
\usepackage{svg}
\usepackage[spaces,hyphens]{url}
\usepackage{hyperref}
\usepackage{url}
\usepackage[numbers]{natbib}
\hypersetup{breaklinks=true}
\usepackage{breakurl}

\theoremstyle{plain}
\theoremstyle{definition}
\newtheorem{assum}{Assumption}

\newcommand{\remove}[1]{}

\begin{document}

\title{Caching Contents with Varying Popularity using Restless Bandits\thanks{This work was supported by Centre for Network Intelligence, Indian Institute of Science (IISc), a CISCO CSR initiative.}}


\author{Pavamana K J \and Chandramani Singh}
\institute{Department of Electronic Systems Engineering, Indian Institute of Science, Bengaluru, India \\
\email{pavamanakj@gmail.com,chandra@iisc.ac.in}}


\maketitle
\begin{abstract}
We study content caching in a wireless network in which the users are connected through a base station that is equipped with a finite capacity cache. We assume a fixed set of contents whose popularity vary with time. Users' requests for the contents depend on their instantaneous popularity levels. Proactively caching contents at the base station incurs a cost but not having requested contents at the base station also incurs a cost. We propose to proactively cache contents at the base station so as to minimize content missing and caching costs. We formulate the problem as a discounted cost Markov decision problem that is a restless multi-armed bandit problem. We provide conditions under which the problem is indexable and also propose a novel approach to manoeuvre a few parameters to render the problem indexable. We demonstrate efficacy of the Whittle index policy via numerical evaluation. 

\remove{
Mobile networks are experiencing prodigious in-
crease in data volume and user density , which exerts a great burden on mobile core networks and backhaul links. An efficient technique to lessen this problem is to use caching i.e. to bring the data closer to the users by making use of the caches of edge network nodes, such as fixed or mobile access points and even user devices. The performance of a caching depends on contents that are cached. In this paper,
we examine the problem of content caching at the wireless edge(i.e. base stations) to minimize the discounted cost incurred over infinite horizon.
We formulate this problem as a restless bandit problem, which is  hard to solve. We begin by showing an optimal policy is of threshold type. Using these structural results, we prove the indexability of the problem, and use Whittle index policy to minimize the discounted cost. Numerical results show that Whittle index policy is close to optimal policy.
}

\end{abstract}

\section{Introduction}

The exponential growth of intelligent devices and mobile applications poses a significant challenge to Internet backhaul, as it struggles to cope with the surge in traffic. According to Cisco's annual Internet report 2023~\cite{Cisco}, approximately two-thirds of the world's population will have access to Internet by 2023, and the number of devices connected to IP networks will exceed three times the global population. This extensive user base will generate a high demand for multimedia content, such as videos and music. However, this increased traffic is often due to repeated transmissions of popular content, leading to an unnecessary burden on the network. The resulting influx of content requests has adverse effects on latency, power consumption, and service quality.


To address these challenges, proactive content caching at the periphery of mobile networks has emerged as a promising solution. By implementing caches at base stations, it becomes possible to pre-store requested content in advance. As a result, content requests can be efficiently served from these local caches instead of remote servers, benefiting both users and network operators. Users experience reduced latency and improved quality of experience when accessing content from intermediary base stations. For network operators, caching content at the network edge significantly reduces network overhead, particularly in cases where multiple users request the same content, such as popular videos and live sports streams.


Notwithstanding the widespread benefits of including
content caching abilities in the networks, there are also several challenges in deploying the caching nodes. First of all, the size of a cache is constrained and caching contents
incurs a cost as well. So, it is not viable to store each and every content that
can possibly be requested by the user in the cache. This calls for efficient strategies to determine the contents that should be stored in the cache. 

In this work, we aim at minimizing the discounted total cost incurred in delivering contents to end users, which consists of content missing and caching costs. We consider a fixed set of contents with varying popularity and a single cache, and design
policies that decide which contents should be cached
so as to minimize the discounted total cost while simultaneously satisfying caching capacity constraint of the base station.

The problem at hand is framed as a Markov decision process (MDP) \cite{puterman2014markov}, resembling a restless multiarmed bandit (RMAB) scenario \cite{Whittle1988restless}. Although value iteration \cite{puterman2014markov,bertsekas1995dynamic} theoretically solves RMAB, it is plagued by the curse of dimensionality and provides limited solution insights. Therefore, it is advantageous to explore less intricate approaches and assess their efficacy. An esteemed strategy for RMAB problems is the Whittle index policy \cite{Whittle1988restless}. This Whittle index policy has been widely employed in the literature and has proven highly effective in practical applications~\cite{glazebrook2002index, glazebrook2006some, ansell2003whittle}.
Whittle~\cite{Whittle1988restless} demonstrated the optimality of index-based policies for the Lagrangian relaxation of the restless bandit problem, introducing the concept of the Whittle index as a useful heuristic for restless bandit problems. Hence, we suggest employing this policy to address the task of optimizing caching efficiently.

\subsection{Related Work}

\subsubsection{Content Caching}
There are two types of caching policies, proactive or reactive. Under a reactive policy, a content can be cached upon the user's request. When a user requests a specific content, the system first checks if the content is available in the local cache. If the content is found in the cache, it is delivered to the user directly from the cache. However, if the content is not present in the cache, the system initiates a process to fetch the content progressively from the server. Li et al.~\cite{li2016popularity} proposed a reactive caching algorithm PopCaching that uses popularity evolution of the contents in determining  which contents to cache.

In proactive caching,  popularity prediction algorithms are used to predict user demand and
to decide which contents are cached and which are evicted.  A. Sadeghi et al. ~\cite{sadeghi2017optimal} proposed an
intelligent proactive caching scheme to cache fixed collection
of contents in an online social network to reduce the energy
expenditure in downloading the contents. Gao et al.~\cite{gao2020design} proposed a dynamic probabilistic caching
for a scenario where contents popularity 
vary with time. N. Abani et al. ~\cite{abani2017proactive}  designed a proactive caching policy that relies on  predictability of the mobility patterns of mobiles to  predict a mobile
device’s next location and to decide which caching nodes should cache which contents. S. Traverso et al. ~\cite{traverso2013temporal}introduced a novel traffic model known as the Shot Noise Model (SNM). This parsimonious model effectively captures the dynamics of content popularity while also accounting for the observed temporal locality present in actual traffic patterns. ElAzzouni et al.~\cite{elazzouni2020predictive} studied the impact of predictive caching on content delivery latency in wireless networks. They establish a predictive multicast and caching concept in which base stations (BSs) in wireless cells proactively multicast popular content for caching and local access by end users.
\subsubsection{Restless Multi-armed Bandit Problems} 
In a restless multi-armed bandit (RMAB) problem, a decision maker must select a subset of ~$M$ arms from $K$ total arms to activate at any given time. The controller has knowledge of the states and costs associated with each arm and aims to minimize the discounted or time-average cost. The state of an arm evolves stochastically based on transition probabilities that depend on whether the bandit is active. Solving an RMAB problem through dynamic programming is computationally challenging, even for moderately sized problems. Whittle~\cite{Whittle1988restless} proposed a heuristic solution known as the Whittle index policy, which addresses a relaxed version of the RMAB problem where ~$M$ arms are only activated on average. This policy calculates the Whittle indices for each arm and activates the ~$M$ arms with the highest indices at each decision epoch. However, determining the Whittle indices for an arm requires satisfying a certain indexability condition, which can be generally difficult to verify.


G. Xiong et al. ~\cite{xiong2022model} have formulated a content caching problem as a RMAB problem with the objective being minimizing the average content serving latency. They established the indexability of the problem and used the Whittle index policy to minimize the average latency.

There are very few works on RMABs with switching costs, e.g., costs associated with switching active arms. J. L. Ny et al. \cite{1656445} considered a RMAB problem with switching costs, but they allow only one bandit to be  active at any time. Incorporating switching costs in RMAB problems makes the states of the bandits multidimensional. This renders
calculation of the Whittle indices much more complex. The literature on multidimensional RMAB is scarce. The main difficulty lies in establishing indexability, i.e., in ordering the states in a multidimensional space. Notable instances are \cite{aalto2015Whittle,anand2018Whittle} and \cite{duran2022Whittle}] in which the authors have derived Whittle indices. But none of them have considered switching cost. We pose the content caching problem  as a RMAB problem with switching costs ( it is called as  caching cost in the context of caching problem) and develop the simple Whittle index policy.

\paragraph*{Organisation}
The rest of the paper is organised as follows. In Section~\ref{sec: System Model }, we present the system model for the proactive content caching problem and formulate the problem as a RMAB. In Section \ref{sec:Whittle-index}, we show that each single arm MDP has a threshold policy as the optimal policy and is indexable. In Section~\ref{sec:nonindexable}, we manoeuvre a few cost parameters to render the modified MDP indexable in a few special cases in which the original MDP is nonindexable. In Section~\ref{sec: Numerical Results}, we show efficacy of the Whittle index policy via numerical evaluation. Finally, we outline future directions in Section~\ref{sec:conclusion}. 

\section{System Model and Caching Problem}
\label{sec: System Model }
In this section, we first present the system model and then pose the optimal caching problem as a discounted cost Markov decision problem.

\subsection{System Model}
\label{sec:system-model}

We consider a wireless network where the users are connected to a single base station~(BS) which in turn is connected to content servers via the core network.
The content providers have a set of $K$ contents, ${\cal C} = \{1,2,\cdots,K\}$ which are of equal size, at the servers. 
The BS has a {\it cache} where it can store up to $M$ contents. We assume a slotted system. Caching decisions are taken at the slot boundaries. We use $a(t) = (a_i(t), i \in {\cal C})$ to denote the caching status of various contents at the beginning of slot $t$; $a_i(t) = 1$ if Content $i$ is cached and $a_i(t) = 0$ otherwise. We let ${\cal A}$ denote the set of feasible status vectors;
\[
{\cal A} = \left\{a \in \{0,1\}^K: \sum_{i \in {\cal C}} a_i \leq M\right\}.
\]

\subsubsection{Content Popularity} We assume that the contents' popularity is reflected in the numbers of requests in a slot and varies over time.
For any content, its popularity evolution may depend on whether it is stored in the BS' cache or not. 
We assume that for any content, say for Content $i$, given its caching status $a_i$, the numbers of requests in successive slots evolve as a discrete time Markov chain as shown in Figure \ref{fig: Transition}.\footnote{There are several instances of content popularity being modelled as Markov chains, e.g., see~\cite{sadeghi2017optimal, sadeghi2019deep, wu2019dynamic}.}

\begin{figure}[!t]
    \centering
    \includegraphics[width=2.5in,height=1.5in]{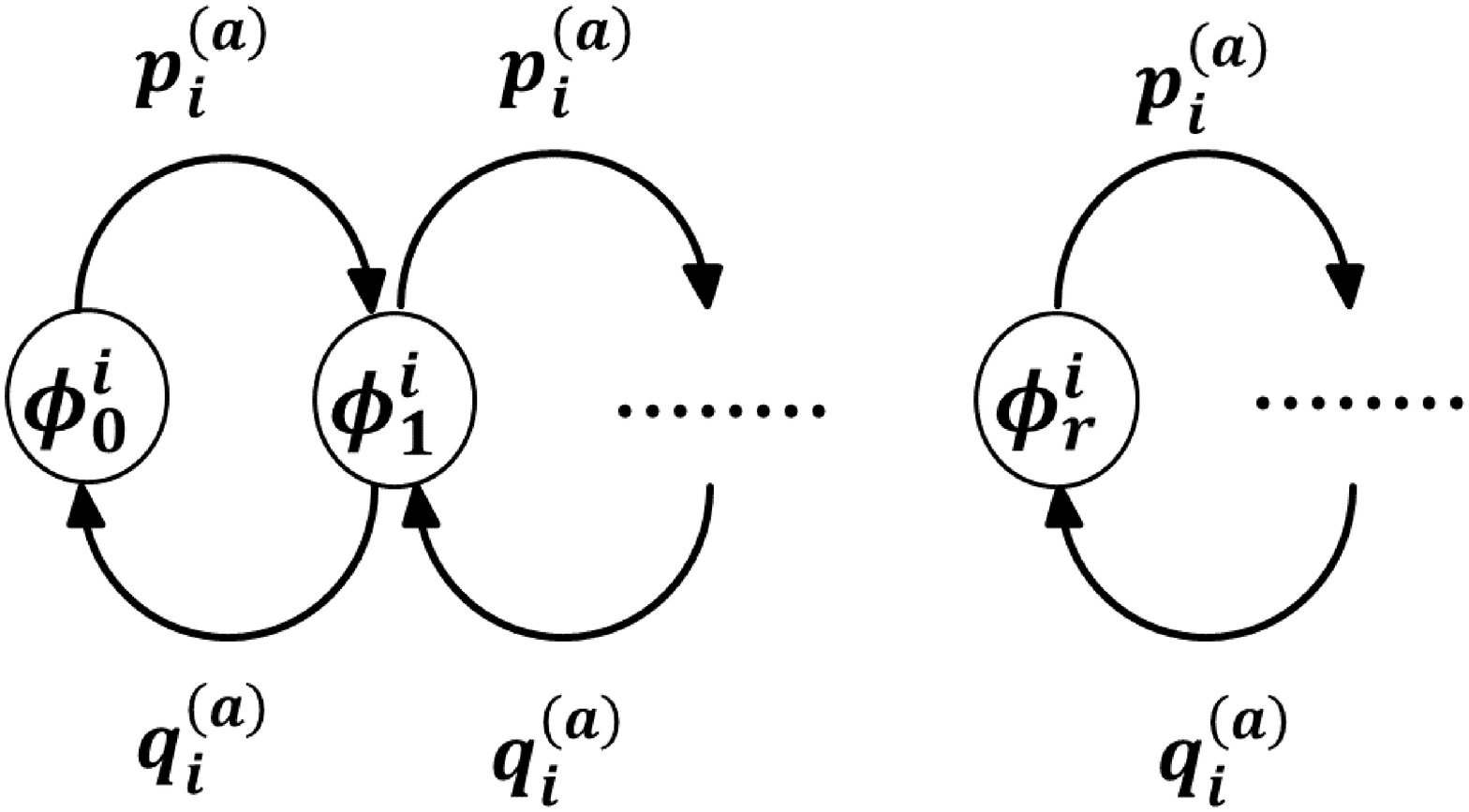}\hfill
    \caption{Popularity evolution of a content. Given its caching status $a \in \{0,1\}$, the average number of requests per slot vary in accordance with these transition probabilities. For clarity, self-loops are not shown.}
    \label{fig: Transition}
\end{figure}

In Figure \ref{fig: Transition}, numbers of requests $\phi^i _{r} \in \mathbb{Z}_+$ for $r = 1,2,\cdots$ and it is an increasing sequence. We do not show self loops for clarity. We also make the following assumption.
\begin{assum}
\label{assum:popularity-stoch-order}
For all $i \in {\cal C}$, $p^{(1)}_i \geq p^{(0)}_i$ and $q^{(1)}_i \leq q^{(0)}_i$. 
\end{assum}
Assumption~\ref{assum:popularity-stoch-order} suggests that, statistically, a content's popularity grows more if it is cached. We need this assumption to establish that optimal caching policy is a {\it threshold policy}.

\subsubsection{Costs} 
We consider the following costs. 
\paragraph*{Content missing cost} If a content, say Content $i$, is not cached at the beginning of a slot and is requested $\phi_r$ times in that slot, a cost $C_i(\phi_r)$ is incurred. For brevity  we write this cost as $C_i(r)$ with a slight abuse of notation. Naturally, the functions $C_i: \mathbb{Z}_+ \to \mathbb{R}_+$ are non-decreasing.
We also make the following assumption.
\begin{assum}
\label{assum:miss-hit-cost-concave}
For all $i \in {\cal C}$,
\begin{enumerate}
    \item $ C_i:\mathbb{Z}_+ \to \mathbb{R}_+$ are non-decreasing,
    \item $ C_i:\mathbb{Z}_+ \to \mathbb{R}_+$ are concave, i.e.,
    \end{enumerate}
\begin{equation}
\label{eqn:miss-hit-cost-concave}
C_i(r+1) - C_i(r)  \leq C_i(r) - C_i(r-1) \ \forall r \geq 1.
\end{equation}
\end{assum}


\paragraph*{Caching cost} At the start of each slot, we have the ability to adjust the caching status of the contents based on their requests in the previous slot. Specifically, we can {\it proactively} cache contents that are currently not cached, while removing other contents to ensure compliance with the cache capacity restrictions. Let's use $d$ to represent the cost associated with fetching content from its server and caching it.

We can express the total expected cost in a slot as a function of the cache status in this and the previous slots, say $a$ and $\bar{a}$, respectively, and the request vector in the previous slot, say $\bar{r}$.
Let $c(\bar{a},\bar{r},a)$ denote this cost. Clearly,
$c(\bar{a},\bar{r},a) = \sum_{i \in {\cal C}}c_i(\bar{a}_i,\bar{r}_i,a_i)$ where

\begin{align}
\label{eqn:single-stage-cost1}
\hspace{-0.15in} c_i(\bar{a}_i,\bar{r}_i,a_i) = & da_i(1-\bar{a}_i) + (1-a_i)\bigg(p^{a_i}C_i(\bar{r}_i+1) +  \nonumber \\
     &  q^{a_i}C_i((\bar{r}_i-1)^+) + (1 -p^{a_i} - q^{a_i})C_i(\bar{r}_i)\bigg). 
\end{align}

We define operators $T^u, u \in \{0,1\}$  for parsimonious presentation. For any function $g: \mathbb{Z}_+ \to \mathbb{R}_+$, for all $r \in \mathbb{Z}_+$,
\[
T^ug(r) = p^u g(r+1) + q^u g((r-1)^+) + (1 -p^u - q^u)g(r).  
\]
We can rewrite~\eqref{eqn:single-stage-cost1} in terms of $T^u$s; 
\begin{equation}
\label{eqn:single-stage-cost2}
c_i(\bar{a}_i,\bar{r}_i,a_i) = da_i(1-\bar{a}_i) + (1-a_i)\bigg(T^{a_i}C_i(\bar{r}_i)\bigg). 
\end{equation}

\subsection{Optimal Caching Problem}
\label{sec:opt-caching}


Our goal is to determine caching decisions that minimize the long-term expected discount cost. To be more precise, when provided with the initial request vector $r$ and cache state $a$, we aim to solve the following problem:
\begin{align} 
\text{Minimize } & \ \mathbb{E}\left[\sum_{t=1}^{\infty}\beta^t c(a(t-1),r(t-1),a(t))\Bigg\vert \substack{a(0) =a, \\ r(0)=r}\right] \label{eqn:objective}\\
\text{subject to } & \ a(t) \in {\cal A} \ \forall \ t \geq 1. \nonumber  
\end{align}
The performance measure~\eqref{eqn:objective} is meaningful when the future costs are less important.\footnote{One can as well consider minimizing expected value of $\sum_{t=1}^{\infty}\sum_{i \in {\cal C}}\beta_i^t c_i(a_i(t-1),r_i(t-1),a_i(t))$. This would model the scenario where the contents have geometrically distributed lifetimes with parameters $\beta_i$s. Our RMAB-based solution continues to apply in this case.} 
\subsubsection{Markov Decision Problem} We formulate the optimal caching problem as a discounted cost Markov decision problem. The slot boundaries are the decision epochs. The state of the system at decision epoch $t$ is given by the tuple $x(t) \coloneqq (a(t-1),r(t-1))$.  We consider $a(t)$ to be the action at decision epoch $t$. Clearly, the state space is ${\cal A} \times \mathbb{Z}_+^K $, and the action space is ${\cal A}$.
From the description of the system in  Section~\ref{sec:system-model}, given state $x(t) = (\bar{a},\bar{r})$ and action $a(t) = a$, the state at decision epoch $t+1$ is $x(t+1) = (a,r)$  where
\begin{equation}
\label{eqn:state-transition}
r_i = \begin{cases}
\bar{r}_i + 1 & \text{w.p. } p^{\bar{a}_i}, \\
\bar{r}_i - 1 & \text{w.p. } q^{\bar{a}_i},\\
\bar{r}_i & \text{w.p. } 1-p^{\bar{a}_i}-q^{\bar{a}_i}.
\end{cases}
\end{equation}
For a state action pair $((\bar{a},\bar{r}),a)$, the expected single state cost is given by $c((\bar{a},\bar{r}),a)$ defined in the previous section~(see~\eqref{eqn:single-stage-cost2}). A policy $\pi$ is a sequence of mappings $\{u^\pi_t, t = 1,2,\cdots\}$ where $u^\pi_t: {\cal A} \times \mathbb{Z}_+^K \to {\cal A}$. The cost of a policy $\pi$ for an initial state $(r,a)$ is 
\begin{equation*}
V^{\pi}(a,r) \coloneqq \mathbb{E}\left[\sum_{t=1}^{\infty}\beta^t c(x(t),u^\pi_t(x(t)))\Big\vert x(0) =(a,r)\right]    
\end{equation*}
Let $\Pi$ be the set of all policies. Then the optimal caching problem is $\min_{\pi \in \Pi}V^{\pi}(a,r)$. 

\begin{definition}({\it Stationary Policies})
A policy $\pi = \{u^\pi_t, t = 1,2,\cdots\}$ is called stationary if $u^\pi_t$ are identical, say $u$, for all $t$. For brevity, we refer to such a policy as the stationary policy $u$. Following~\cite[Vol 2, Chapter 1]{bertsekas1995dynamic}, the content caching problem assumes an optimal stationary policy.   
\end{definition}

\paragraph*{Restless multi-armed bandit formulation}
The Markov decision problem described above presents a challenge due to its high dimensionality. However, we can make the following observations:
\begin{enumerate}
    \item The evolution of the popularities of the contents is independent when considering their caching statuses. Their popularity changes are connected solely through the caching actions, specifically the capacity constraint of the cache.
    \item The total cost can be divided into individual costs associated with each content. 
\end{enumerate}
We thus see that the optimal caching problem is an instance of the {\it restless multi-armed bandit problem}~(RMAB) with each arm representing content. We show in Section \ref{sec:Whittle-index} that this problem is {\it indexable}. This allows us to develop a Whittle index policy for the joint caching problem.

By recognizing the similarities between the optimal caching problem and the {\it restless multi-armed bandit problem}~(RMAB), where each arm corresponds to a content, we can conclude that the optimal caching problem can be framed as an RMAB instance. In Section~\ref{sec:Whittle-index}, we establish that this problem is {\it indexable}, enabling us to devise a Whittle index policy to tackle the joint caching problem efficiently.

\begin{remark}
    In practice the popularity evolution could be unknown. The authors in \cite{xiong2022model, avrachenkov2022whittle,fu2019towards, robledo2022qwi} have proposed reinforcement learning~(RL) based approaches to learn Whittle indices in case of Markov chains with unknown dynamics. But these works consider only one-dimensional Markov chains. Thanks to the switching costs, we have a two-dimensional Markov chain at our disposal which renders  convergence of the RL algorithms and calculation of the Whittle indices much
    harder. This constitutes our future work.
\end{remark}
\section{Whittle Index Policy}
\label{sec:Whittle-index}
We outline our approach here. We first solve certain caching problems associated with each of the  contents. We argue that these problems are indexable. Under indexability, the solution to the caching problem corresponding to a content yields Whittle indices for all the states of this content. The Whittle index  measures how rewarding it is to cache that content at that particular state. The Whittle index policy  chooses those $M$ arms whose current states have the largest Whittle indices and, among these caches, those with positive Whittle indices.

\subsection{Single Content Caching Problem}
\label{sec:single-content-problem}
We consider a Markov decision problem associated with Content $i$. Its state space is $\{0,1\} \times \mathbb{Z}_+$ and its action space is $\{0,1\}$. Given $x_i(t) = (\bar{a}_i,\bar{r}_i)$ and action $a_i(t) = a_i$, the state evolves as described in Section \ref{sec: System Model } (see~\eqref{eqn:state-transition}).
The expected single-stage cost is 
\begin{equation}
\label{eqn:single-stage-cost3}
c_{i,\lambda}(\bar{a}_i,\bar{r}_i,a_i) = \lambda a_i + da_i(1-\bar{a}_i) + (1-a_i)T^{a_i}C_i(r_i).
\end{equation}
Observe that a constant penalty $\lambda$
is incurred in each slot in which Content $i$ is stored in the cache.
Here a policy $\pi$ is a sequence $\{u^\pi_t, t = 1,2,\cdots\}$ where $u^\pi_t: \{0,1\} \times \mathbb{Z}_+ \to \{0,1\}$. Given initial state $(a_i,r_i)$ the problem minimizes 
\begin{equation*}
V^{\pi}_{i,\lambda}(a_i,r_i) \coloneqq  \mathbb{E}\left[\sum_{t=1}^{\infty}\beta^t c_{i,\lambda}(x_i(t),u^\pi_t(x_i(t)))\Big\vert x_i(0) =(a_i,r_i)\right] 
\end{equation*}
over all the policies to yield the optimal cost function $V_{\lambda}(a_i,r_i) \coloneqq \min_{\pi}V^{\pi}_{i,\lambda}(a_i,r_i)$. We analyze this problem below. However, we omit the index $i$ for brevity. 
The single content caching problem also  assumes an optimal stationary policy.   
Moreover, following~\cite[Vol 2, Chapter 1]{bertsekas1995dynamic}, the optimal cost function $V_{\lambda}(\cdot,\cdot)$ satisfies the following Bellman's equation.
\begin{equation}
\label{eqn:val-iteration-v-NonVI}
V_{\lambda}(a,r) = \min_{a'\in\{0,1\}} Q_{\lambda}(a,r,a'),
\end{equation}
where
\begin{equation}
\label{eqn:Q-iteration-v-NonVI}
Q_{\lambda}(a,r,a') \coloneqq c_{\lambda}(a,r,a')+\beta T^{a'}V_{\lambda}(a',r).
\end{equation}

Here $T^{a'}V_{\lambda}(a',r)$ is defined 
by applying operator $T^{a'}$ to the function $V_{\lambda}(a',\cdot):\mathbb{Z}_+ \to \mathbb{R}_+$. The set ${\cal P}(\lambda)$ of the states in which action $0$ is optimal, referred to as the {\it passive set}, is given by
\begin{equation*}
\label{eqn:passve-set}
{\cal P}(\lambda) = \{(a,r): Q_{\lambda}(a,r,0) \leq Q_{\lambda}(a,r,1) \}.
\end{equation*}
The complement of ${\cal P}(\lambda)$ is referred to as the {\it active set}.
Obviously, the penalty $\lambda$
influences the partition of the passive and
active sets.
\begin{definition}({\it Indexability})
An arm is called indexable if the passive set ${\cal P}(\lambda)$ of
the corresponding content caching problem monotonically increases from $\emptyset$ to the whole state space $\{0,1\} \times \mathbb{Z}_+$ as the penalty $\lambda$ increases from $-\infty$ to $\infty$. An RMAB is called indexable if it's every arm is indexable~\cite{Whittle1988restless}.
\end{definition}
The minimum penalty needed to move a state from the active set to the passive
set measures how attractive this state is. This motivates the following definition of the Whittle index. 
\begin{definition}({\it Whittle index})
If an arm is indexable, its Whittle index $w(a,r)$ associated with state $(a,r)$ is the minimum penalty that moves this state from the active set to the passive set. 
Equivalently, 
\begin{equation}
\label{eqn:Whittle-index}
w(a,r) = \min \{\lambda: (a,r) \in {\cal P}(\lambda)\}.
\end{equation}
\end{definition}
Before we establish the indexability of the arm~(content) under consideration, we define threshold policies and show that 
the optimal policy for the single content caching problem is a threshold policy.
\begin{definition}({\it Threshold Policies})
A stationary policy $u$ is called a threshold policy if it is of the form
\begin{equation*}
u(a,r) = \begin{cases}
            0 \text{ if } r \leq r^a, \\
            1 \text{ otherwise,}
            \end{cases}
            \end{equation*}
for some $r^a \in \mathbb{Z}_+, a = 0,1$. In the following, we refer to such a policy as the threshold policy $(r^0,r^1)$.
\end{definition}

\subsection{Optimality of a Threshold Policy}
Observe that the optimal cost function $V_\lambda(\cdot,\cdot)$ is obtained as the limit of the following {\it value iteration}~\cite[Vol 2, Chapter 1]{bertsekas1995dynamic}. For all $(a,r) \in \{0,1\} \times \mathbb{Z}_+$, $V^0_{\lambda}(a,r) = 0$, and for $n \geq 1$,
\begin{equation}
\label{eqn:val-iteration-v}
V^n_{\lambda}(a,r) = \min_{a'\in\{0,1\}} Q^n_{\lambda}(a,r,a'),
\end{equation}
where
\begin{equation}
\label{eqn:Q-iteration-v}
Q^n_{\lambda}(a,r,a') \coloneqq c_{\lambda}(a,r,a')+\beta T^{a'}V^{n-1}_{\lambda}(a',r).
\end{equation}
We start by arguing that $V_{\lambda}^n(a,r)$ is concave and increasing in $r$ for all $a$ and $n \geq 0$. But it requires the following assumption.

\begin{assum}
\label{assum:transition-prob}
\begin{equation*}
    p^0\bigg( C(3) - C(2) \bigg) - ( 2p^0 + q^0 -1 )\bigg( C(2) - C(1 )\bigg) + ( p^0 + 2 q^0 -1)\bigg(  C(1) - C(0) \bigg) \leq 0.   
\end{equation*}
\end{assum}

\begin{lemma}
\label{lem:value-fun-concave}
$V_{\lambda}^n(a,r)$ is concave and non-decreasing in $r$ for all $a$ and $n \geq 0$.
\end{lemma}
\begin{proof}
Please refer to our extended version~\cite{j2023caching}
\end{proof}

\begin{remark}
We require Assumption~\ref{assum:transition-prob} merely to prove that $V_{\lambda}^1(a,r)$ is concave. It follows via induction that $V_{\lambda}^n(a,r), n \geq 2$ are also concave.
\end{remark}
\begin{lemma} For all $n \geq 1$,
\label{lem:value-fun-properties}
 \begin{enumerate}
\item $Q^n_{\lambda}(a,r,0)-Q^n_{\lambda}(a,r,1)$ are non-decreasing in $r$ for $a = 0,1$.
\item $V_{\lambda}^n(0,r) -  V_{\lambda}^n(1,r)$ are non-decreasing in $r$. 
\end{enumerate}
\end{lemma}
\begin{proof}
Please refer to our extended version~\cite{j2023caching}
\end{proof}

The following theorem uses Lemma~\ref{lem:value-fun-concave} and~\ref{lem:value-fun-properties} to establish that there exist optimal threshold policies for the single content caching problems.
\begin{theorem}
\label{thm:threshold-policy}
For each $\lambda \in \mathbb{R}$ there exist $r^0(\lambda),r^1(\lambda) \in \mathbb{Z}_+$ such that the threshold policy $(r^0(\lambda),r^1(\lambda))$ is an optimal policy for the single content caching problem with penalty $\lambda$. Also, $r^0(\lambda) \geq r^1(\lambda)$. 
\end{theorem}

\begin{proof}
Please refer to our extended version~\cite{j2023caching}
\end{proof}

\subsection{Indexability of the RMAB}
We now exploit the existence of an optimal threshold policy to argue that the RMAB formulation of the content caching problem is indexable.

\begin{lemma} For all $n \geq 1$,
\label{lem:value-fun-lambda}
 \begin{enumerate}
\item $Q^n_{\lambda}(a,r,1)-Q^n_{\lambda}(a,r,0)$ are non-decreasing in $\lambda$ for $a = 0,1$.
\item $V_{\lambda}^n(1,r) -  V_{\lambda}^n(0,r)$ are non-decreasing in $\lambda$. 
\item $V_{\lambda}^n(a,r+1) -  V_{\lambda}^n(a,r)$ are non-decreasing in $\lambda$ for $a = 0,1$.
\end{enumerate}
\end{lemma}
\begin{proof}
Please refer to our extended version~\cite{j2023caching}
\end{proof}
\begin{theorem}
\label{thm:indexability}
Under Assumptions~\ref{assum:popularity-stoch-order},~\ref{assum:miss-hit-cost-concave} and~\ref{assum:transition-prob} the content caching problem is indexable. \end{theorem}

\begin{proof}
Please refer to our extended version~\cite{j2023caching}
\end{proof}

\begin{remark}
A more common approach to show indexibility of a RMAB have been arguing that the value function is convex, e.g., see~\cite{ansell2003whittle,larranaga2014index}. But it can be easily verified that the value functions in our problem will not be convex even if the content missing costs are assumed to be convex.    
\end{remark}
\begin{remark}
Theorems~\ref{thm:threshold-policy} and~\ref{thm:indexability} require Assumption~\ref{assum:transition-prob}. Many works in literature have relied on such conditions on transition probabilities for indexibility of RMABs~(see~\cite{liu-zhao10indexibility,meshram-etal18whittle-index}).
\end{remark}

\subsection{Whittle Index Policy for the RMAB}
We now describe the Whittle index policy for the joint content caching problem.
As stated earlier, it  chooses those $M$ arms whose current states have the largest Whittle indices and among these caches the ones with positive Whittle indices. It is a stationary policy. Let  $u^W:{\cal A} \times \mathbb{Z}_+^K   \to {\cal A}$ denote this policy. Then 
\[
u^W_i(a,r) = \begin{cases}
1 \ \text{if $w(a_i,r_i)$ is among the highest $M$ values} \\
\ \ \text{in $\{w(a_i,r_i), i \in {\cal C}\}$ and $w(a_i,r_i) > 0$},  \\
0 \text{ otherwise.}
\end{cases}
\]

\paragraph*{Holding cost}
There can also be a holding cost for keeping a content in the cache. Let $h$ denote this fixed holding cost per content per slot. The  content caching problem remains unchanged except that single stage cost associated with Content $i$ becomes
\begin{equation*}
c_i(\bar{a}_i,\bar{r}_i,a_i) = ha_i + da_i(1-\bar{a}_i) + (1-a_i)T^{a_i}C_i(r_i).
\end{equation*}
Comparing it with~\eqref{eqn:single-stage-cost3}, we see that the penalty $\lambda$ can be interpreted as the fixed holding cost.
Obviously, in the presence of the holding cost, the content caching problem can be solved following the same approach as above.

\begin{table}[t]
\small
\centering
\caption{Choices of $\hat{C}(0)$ and $\hat{C}(1)$ that render the problem indexable; empty cells indicate absence of suitable choices.}
\label{table:modified-costs}
\begin{tabular}{p{0.3\linewidth}|p{0.6\linewidth}}
\hline
Case & Costs \\
\hline
\\[-0.7em]
$p^0 + 2q^0 \leq 1 $ & $\hat{C}(1)=C(1)$, $\hat{C}(0)=\min\left\{C(0), C(1) - \frac{F}{p^0 + 2 q^0 -1}\right\}$ \\
\hline
\\[-0.7em]
$p^0 + 2q^0 > 1,  2p^0 + q^0 < 1 $ & \\
\hline
\\[-0.7em]
$p^0 + 2q^0 > 1, 2p^0 + q^0 \geq 1, q^0 \leq p^0$ & $\hat{C}(1) = C(2) - \frac{p^0(C(3) - C(2))}{p^0 - q^0}$, $\hat{C}(0) = 2 \hat{C}(1) - C(2)$ \\
\hline
\\[-0.7em]
$p^0 + 2q^0 > 1, 2p^0 + q^0 \geq 1, q^0 > p^0$ & \\
\hline
\multicolumn{2}{l}{* where $F = p^0\bigg((C(2) -C(1))-(C(3) - C(2))\bigg)- (1-p^0-q^0)\bigg(C(2) - C(1)\bigg)$.}
\end{tabular}
\end{table}

\section{Noncompliance with Assumption~\ref{assum:transition-prob}}
\label{sec:nonindexable}

We have so far assumed that the costs and the transition probabilities satisfy Assumption~\ref{assum:transition-prob} to ensure indexability of the content caching problem. We now explore a novel approach of manoeuvring the content missing costs so that the modified  content caching problem is ``close'' to original problem and is indexable. We obtain the Whittle index policy for the modified problem and use it for the original problem.

More specifically, let us consider a particular content for which Assumption~\ref{assum:transition-prob} is not met. We investigate the possibility of tinkering only $C(0)$ and $C(1)$ to achieve indexability.\footnote{As in Section~\ref{sec:single-content-problem}, we omit the content index.} In other words, we consider content missing costs  $\hat{C}:\mathbb{Z}_+ \to \mathbb{R}_+$ with $\hat{C}(r) = C(r), r \geq 2$ and other costs and popularity evolution also unchanged. We demonstrate that in certain special cases adequate choices of $\hat{C}(0)$ and $\hat{C}(1)$ render the modified problem indexable. Our findings are summarized in  Lemma~\ref{lem:modified-indexability} and Table~\ref{table:modified-costs}.  

\begin{lemma}
\label{lem:modified-indexability}
If (a)~$p^0 + 2q^0 \leq 1$ or (b)~$p^0 + 2q^0 > 1, 2p^0 + q^0 > 1, q^0 \leq p^0$ then, with $\hat{C}(0)$ and $\hat{C}(1)$ as in Table~\ref{table:modified-costs},  
\begin{enumerate}
    \item $\hat{C}:\mathbb{Z}_+ \to \mathbb{R}_+$ is non-decreasing,
    \item $\hat{C}:\mathbb{Z}_+ \to \mathbb{R}_+$ is concave,
    \item the modified costs satisfy Assumption~\ref{assum:transition-prob}.
\end{enumerate}
In other cases, there do not exist 
$\hat{C}(0)$ and $\hat{C}(1)$ such that $\hat{C}:\mathbb{Z}_+ \to \mathbb{R}_+$ satisfies all these three properties.
\end{lemma}
\begin{proof}
Please refer to our extended version~\cite{j2023caching}
\end{proof}

Typically the number of requests remains much higher than $0$ or $1$, and so the optimal caching policy and the cost for the amended problem are close to those for the original problem. Consequently, the  Whittle index policy for the modified problem also performs well for the original problem. We demonstrate it in Section~\ref{sec: Numerical Results} though theoretical performance bounds have eluded us so far.

\section{Popularity evolution oblivious to the caching action} 

If the popularity evolution is independent of the caching action,  $p^0 = p^1$ and $ q^0 = q^1$. In this case one can easily show via induction that  $V_{\lambda}(1,r) - V_{\lambda}(0,r)$ is non-decreasing in  $\lambda$, and 
 consequently, the associated restless bandit problem is indexable. The proofs of these assertions follow from similar steps as in the proofs of the second part of 
 Lemma~\ref{lem:value-fun-lambda} and of 
 Theorem~\ref{thm:indexability}, respectively.

\begin{remark}
    When popularity evolution does not depend on the caching action the above indexability assertion does not rely on the the special transition structure of Figure~\ref{fig: Transition}. In other words, the restless bandit problems corresponding to the caching problems with arbitrary Markovian popularity evolutions are indexable as long as the evolutions do not depend on the caching action.  
\end{remark}

Below we generalize the above assertion by deriving a condition on the difference of popularity evolution with and without caching for indexability. Towards this, we define $\delta \coloneqq 2 \max\{ |p^0 - p^1| ,|q^0 - q^1| \}$.
We prove the following indexability result.

\begin{theorem}
\label{thm:indexability 2}
If $\beta \leq \max\{\frac{1}{1+\delta},\frac{1}{2}\}$, the content caching problem is indexable. 
\end{theorem}
\begin{proof}
We argue that
\begin{enumerate}
\item $Q^1_{\lambda}(a,r, 1) - Q^1_{\lambda}(a,r, 0) $ is nondecreasing in $\lambda$,
\item If $Q^n_{\lambda}(a,r,1)-Q^n_{\lambda}(a,r,0)$ for $a = 0,1$ are nondecreasing in $\lambda$, so is $V_{\lambda}^n(1,r) -  V_{\lambda}^n(0,r)$,
\item If $V_{\lambda}^n(1,r) -  V_{\lambda}^n(0,r)$ is nondecreasing in $\lambda$ and $\beta \leq \max\{\frac{1}{1+\delta},\frac{1}{2}\}$, then $Q^{n+1}_{\lambda}(a,r, 1) - Q^{n+1}_{\lambda}(a,r, 0) $ is also nondecreasing in $\lambda$. 
\end{enumerate}
Hence, via induction, $Q_{\lambda}(a,r, 1) - Q_{\lambda}(a,r, 0) $ is nondecreasing in $\lambda$
which implies that the problem is indexable. 
Please see~\cite{j2023caching} for the details.
\end{proof}

\remove{
In this section, we show that the single content caching problem  with $ p^0 \approx p^1$ and $ q^0 \approx q^1$ is still indexable.
\begin{assum} \label{assum:increasing}
   
   \begin{enumerate}
       \item $ \frac{C(r+1) - C(r)}{C(r) - C(r-1)}$ is increasing in $r$.
       \item  $\delta = 2 \max\{ |p^0 - p^1| ,|q^0 - q^1| \}$
       \item $\beta \leq \epsilon =\frac{C(2) - C(1)}{C(1)- C(0)} $
       \item $\delta \leq  \min\bigg\{ \frac{(\epsilon - \beta) \epsilon^2}{3 \beta} , \frac{1-\beta}{3 \beta} \bigg\} $
   \end{enumerate}
\end{assum}

\begin{theorem}
\label{thm:threshold-policy 2}
For each $\lambda \in \mathbb{R}$ there exist $r^0(\lambda),r^1(\lambda) \in \mathbb{Z}_+$ such that the threshold policy $(r^0(\lambda),r^1(\lambda))$ is an optimal policy for the single content caching problem with penalty $\lambda$ under assumptions ~\ref{assum:miss-hit-cost-concave} and ~\ref{assum:increasing}.
\begin{proof}
Please refer to our extended version~\cite{j2023caching}
\end{proof}

\begin{theorem}
\label{thm:indexability 2}
Under Assumptions~\ref{assum:miss-hit-cost-concave} and~\ref{assum:increasing} the content caching problem is indexable. \end{theorem}
\end{theorem}
\begin{proof}
Please refer to our extended version~\cite{j2023caching}
\end{proof}
}

Observe that, for $\beta \leq 1/2$, the problem is indexable for arbitrary $p^0,p^1,q^0$ and $q^1$.
For $\beta \in (1/2,1)$, the problem is indexable if $\delta \leq \frac{1-\beta}{\beta}$. 
 Assumptions~\ref{assum:popularity-stoch-order}, \ref{assum:miss-hit-cost-concave} or \ref{assum:transition-prob}
are not needed in these cases.

\begin{remark}
 Meshram et al.~\cite{meshram-etal18whittle-index} have considered  a restless single-armed hidden Markov bandit with
two states and have shown that under certain conditions on transition probablities, the arm is indexable for $\beta \in (0, 1/3)$. More recently, Akbarzadeh and Mahajan~\cite{akbarzadeh2020conditions} have provided sufficient conditions for indexability of general restless multiarmed bandit problems. One can directly infer from~\cite[Theorem 1 and Proposition 2]{akbarzadeh2020conditions} that the caching problem is indexable for $\beta \leq 1/2$. On the other hand, Theorem~\ref{thm:indexability 2} implies that the caching problem can be indexable even for $\beta \in (1/2,1)$. In Section~\ref{sec: Numerical Results}, we  numerically show that the caching problem can be indexable even when $\beta > 1/(1+\delta)$.    
\end{remark}

\section{Numerical Results}
\label{sec: Numerical Results}

In this section, we numerically evaluate the Whittle index policy for a range of system parameters. We demonstrate a variation of Whittle indices with various parameters. We also compare the performance of the Whittle index policy with those of the optimal and greedy policies.  We compute Whittle indices using an algorithm proposed in~\cite{Gast_2023}. We assume $C_i(r) = 3\sqrt{r}$ for all $i \in {\cal C}$. We use $ p^{a_i}, q^{a_i}, a \in \{0,1\}$, for all $i \in {\cal C}$ which satisfy  Assumptions~\ref{assum:popularity-stoch-order} and~\ref{assum:transition-prob}.

\paragraph*{Whittle indices for different caching costs}
We consider $K = 40$ contents and a cache size $M = 16$. We assume transition probabilities $p^0 = 0.06082 , q^0= 0.38181, p^1 = 0.63253 , q^1 = 0.26173$ for all the contents and discount factor $\beta = 0.95$. We plot the Whittle indices for two different values of  the caching costs $d = 10$ and $d = 400$ in Figures~\ref{fig:my_label}(a) and~\ref{fig:my_label}(b), respectively. As expected, $w(0,r)$ and  $w(1,r)$ are increasing in $r$. We also observe that, for a fixed $r$, $w(1,r)$ does not vary much as the caching   cost is changed but $w(0,r)$ decreases with   the caching cost. This is expected as increasing caching costs makes the passive action~(not caching) more attractive. 

\begin{figure}
\vspace{-0.2in}
\centering
\begin{subfigure}[t]{.45\textwidth}
    \centering
    \includegraphics[width= 2.2in,height=1.5in]{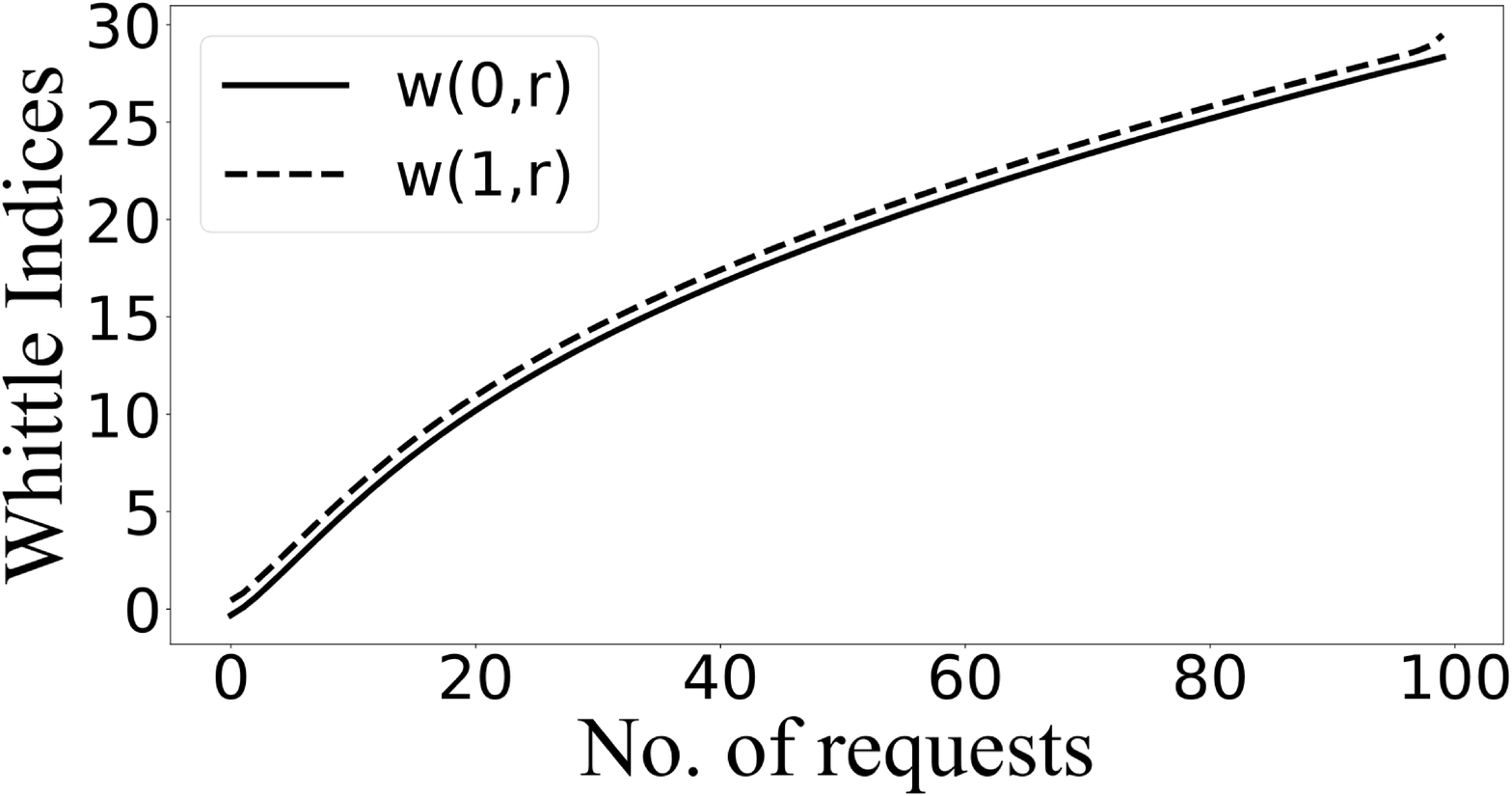}
     \caption{Caching cost $d = 10$}
\end{subfigure}
\hfill
\begin{subfigure}[t]{.45\textwidth}
    \centering
    \includegraphics[width= 2.2in,height=1.5in]{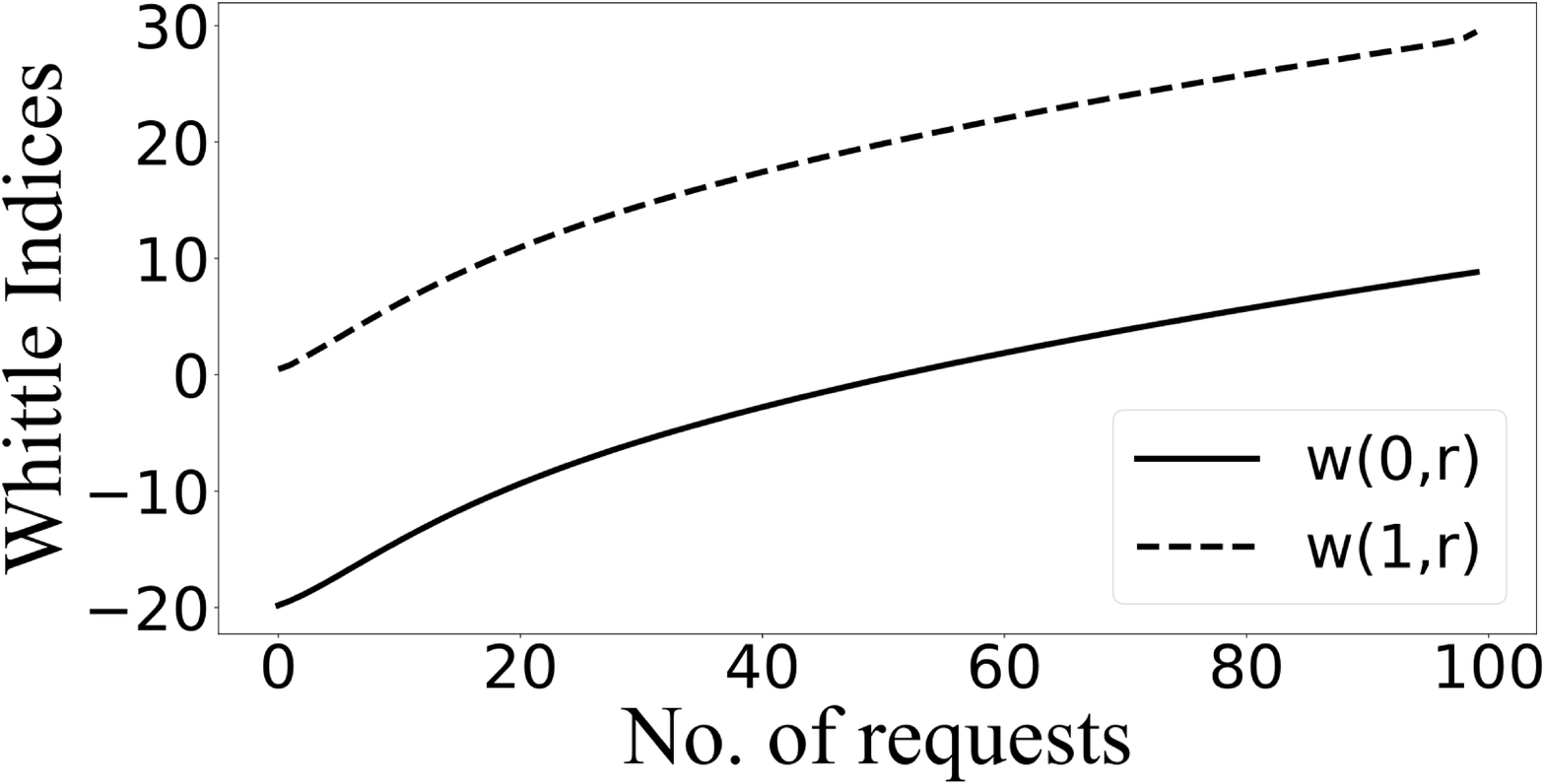}
     \caption{Caching cost $d = 400$}
\end{subfigure}
\caption{Whittle indices for different caching costs}
    \label{fig:my_label}
\vspace{-0.3in}
\end{figure}

\paragraph*{Whittle indices for  different discount factors} We plot the Whittle indices for two different values of  discount factors $\beta = 0.3$ and $\beta = 0.9$ in Figures~\ref{fig:beta}(a) and~\ref{fig:beta}(b), respectively. Other parameters are the same as those for Figure~\ref{fig:my_label}  except caching cost $ d = 10$. Here also, we observe that, for a fixed $r$, $w(1,r)$ does not vary much as the $\beta$ is changed, but $w(0,r)$ increases with   $\beta$.

\begin{figure}
    \centering
    \begin{subfigure}[t]{.45\textwidth}
    \centering
    \includegraphics[width=2.2in,height=1.5in]{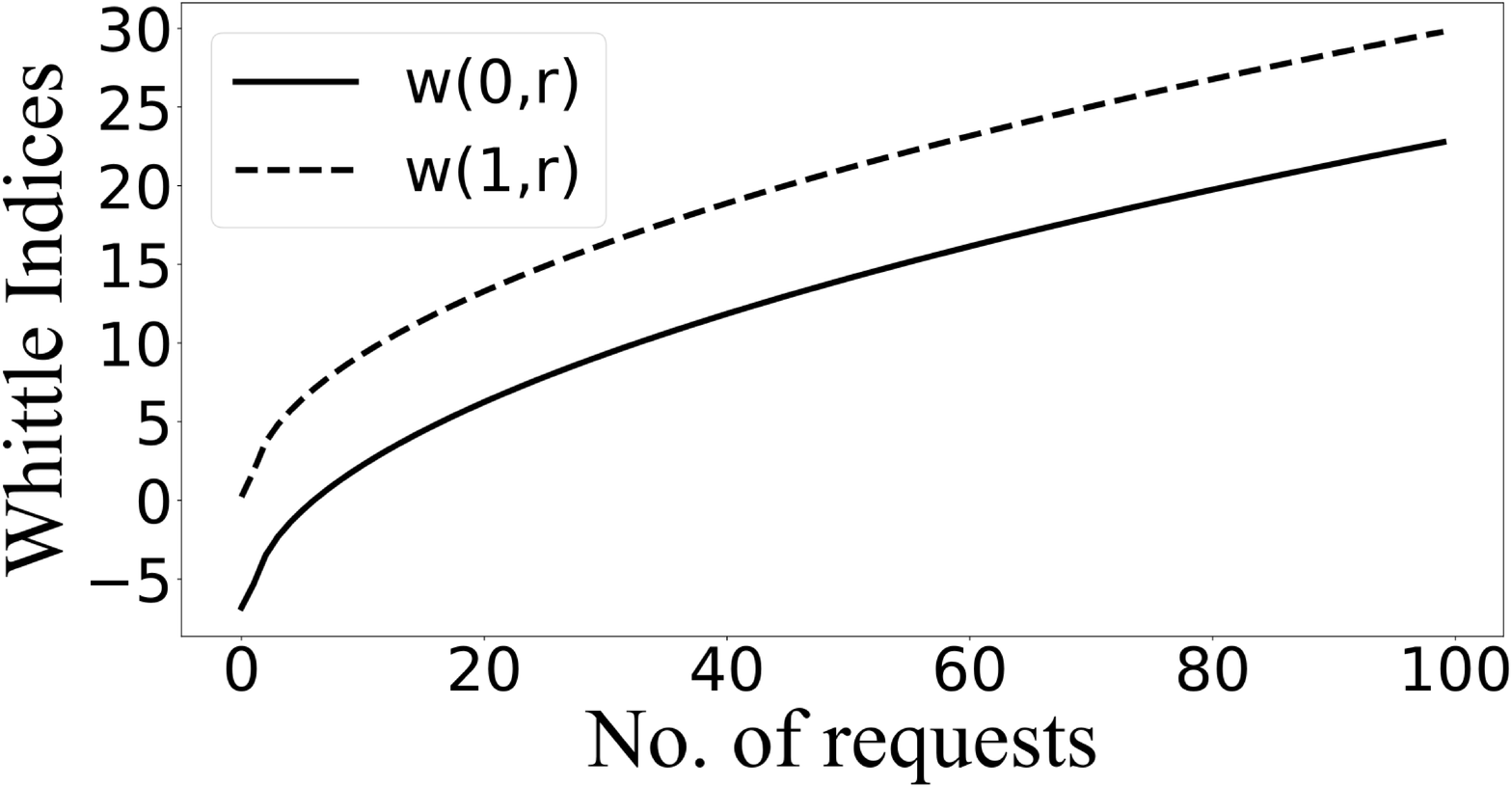}
     \caption{$\beta$ = 0.3}
\end{subfigure}
\hfill
\begin{subfigure}[t]{.45\textwidth}
    \centering
        \includegraphics[width=2.2in,height=1.5in]{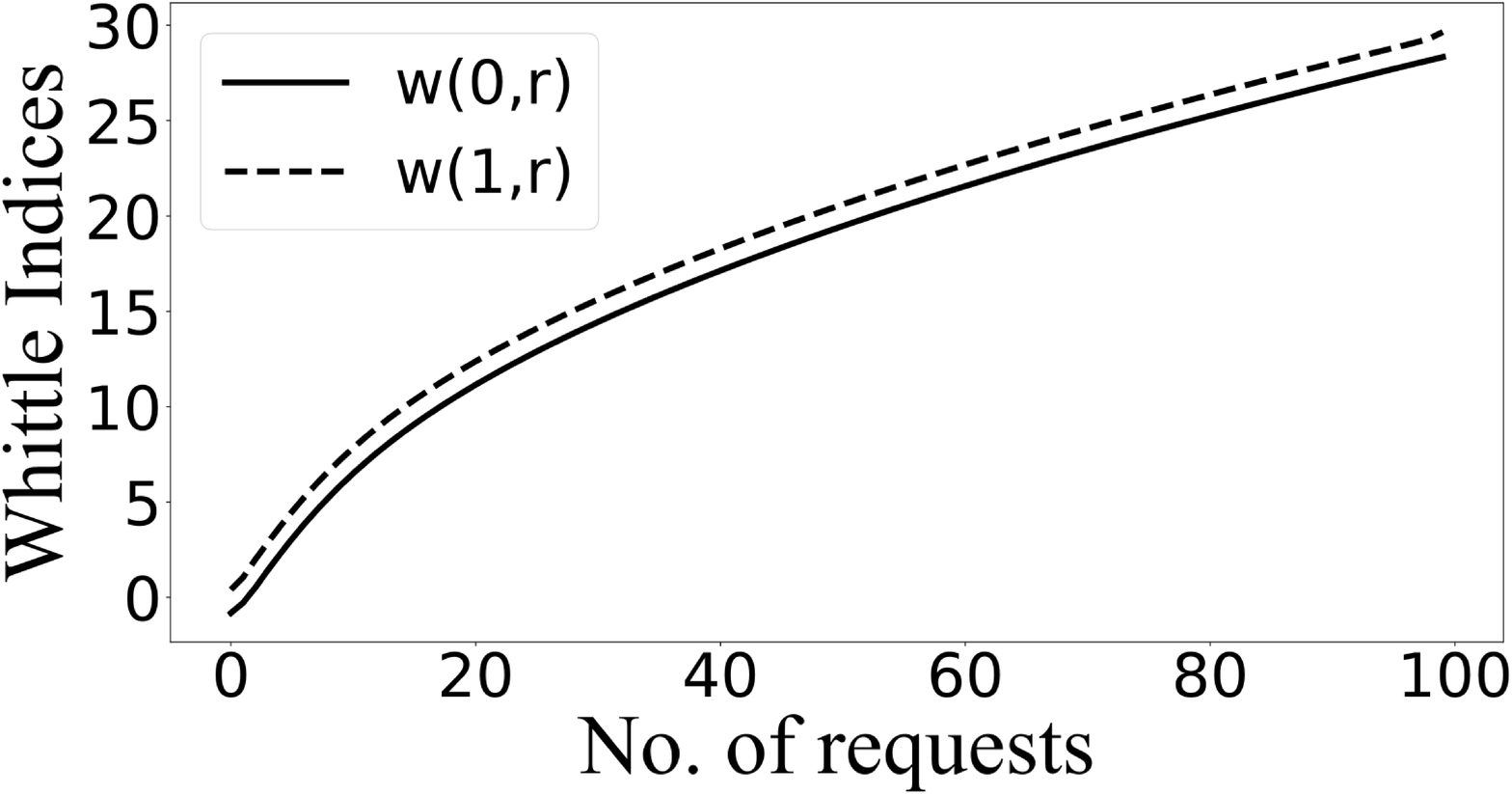}
     \caption{$\beta$ = 0.9}
\end{subfigure}
    \caption{ Whittle indices for different discount factors}
    \label{fig:beta}
\vspace{-0.2in}
\end{figure}

\paragraph*{Whittle indices for  different transition probabilities} We plot the Whittle indices for two different sets of transition probabilities 
\begin{enumerate}
    \item $ p^1 > q^1$ and $p^0 > q^0$ e.g., $p^0 = 0.0093, q^0= 0.0061, p^1 = 0.3274 , q^1 = 0.0030$
    \item $p^1 < q^1$ and $p^0 < q^0$ e.g., $p^0 = 0.0007 , q^0= 0.9362, p^1 = 0.0021 , q^1 = 0.8902  $
\end{enumerate}
Assuming caching cost $ d = 10$ and $\beta = 0.95$, we note that in the first scenario, whether the content is cached or not (active or passive action), the number of requests is likely to increase. Similarly, in the second scenario, the number of requests is likely to decrease regardless of caching status. Fig.~\ref{fig: Different TPM} illustrates the Whittle indices for both cases. It is evident that the Whittle indices are higher in the first case compared to the second case. This outcome is expected since when the number of requests is more likely to increase, caching is anticipated. The Whittle index policy selects the M contents with the highest current state Whittle indices and caches those with positive indices. Thus, as the likelihood of request increases for content to be cached, the corresponding Whittle indices are expected to rise.


\paragraph*{Performance of the Whittle index policy for a problem conforming to Assumptions~\ref{assum:popularity-stoch-order} and~\ref{assum:transition-prob}} 

In our comparison, we evaluate the performance of the Whittle Index Policy against the optimal policy obtained through value iteration and the greedy policy. The greedy policy selects the action that minimizes the total cost outlined in Section~\ref{sec: System Model }, considering all possible actions while satisfying the constraints at each moment. However, the optimal policy is only feasible for small values of $ K, M$. Therefore, we set $K = 3$ and $M = 1$ for our analysis. The parameters $\beta = 0.95$, caching cost $d = 10$, and the probabilities $ p^{a_i}, q^{a_i}, a \in \{0,1\}$, for all $i \in {\cal C}$, are chosen to satisfy Assumptions~\ref{assum:popularity-stoch-order} and~\ref{assum:transition-prob}. Our findings indicate that the Whittle index policy outperforms the greedy policy and approaches the performance of the optimal policy, as depicted in Fig.~\ref{fig: Optimal}


\begin{figure}
    \centering
    \begin{minipage}[t]{0.45\textwidth}
        \centering
        \includegraphics[width=2.2in,height=1.5in]{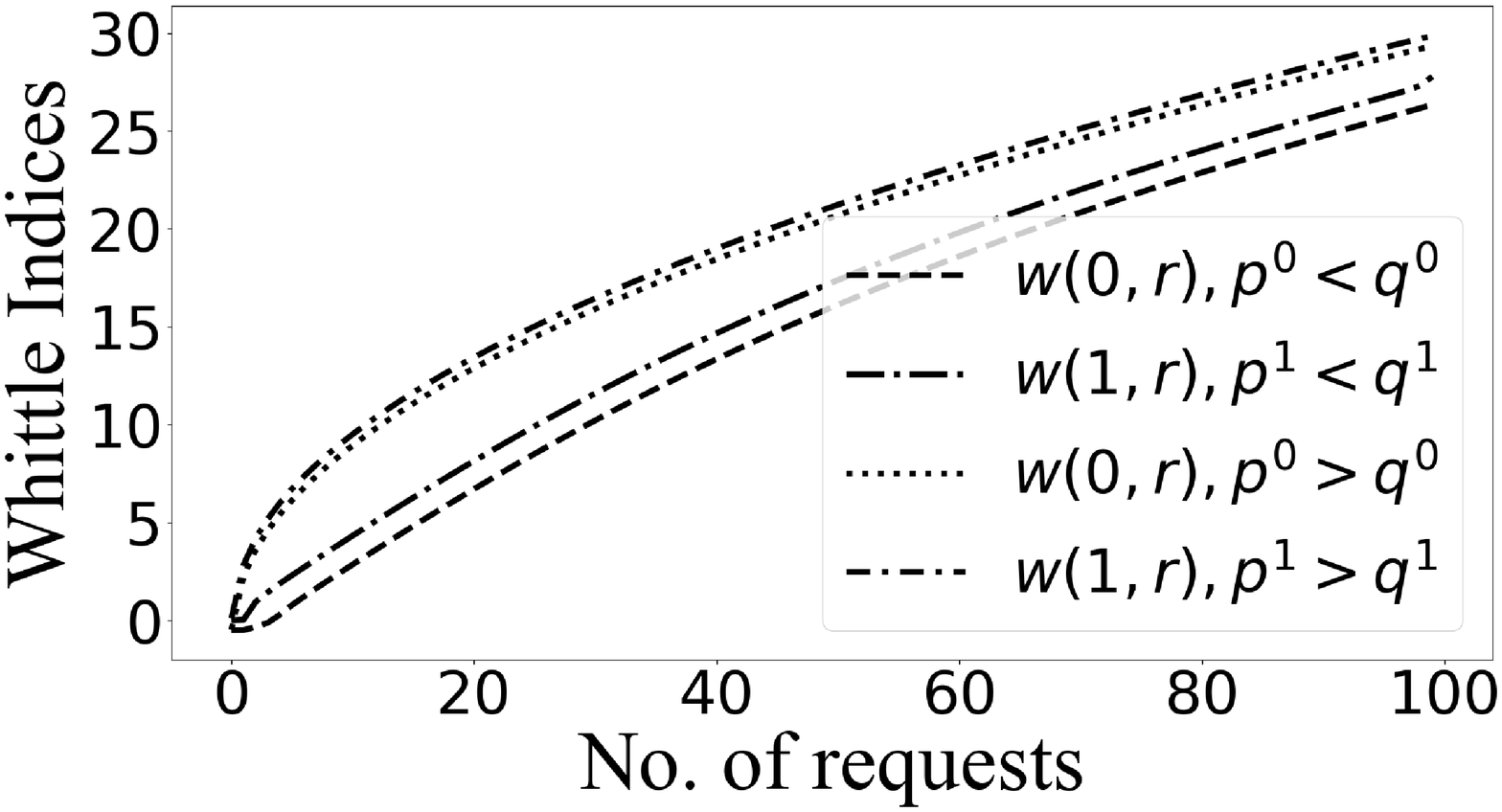}
         \caption{Whittle indices for different transition probabilities}
        \label{fig: Different TPM}
    \end{minipage}\hfill
    \begin{minipage}[t]{0.45\textwidth}
        \centering
        \includegraphics[width=2.2in,height=1.5in]{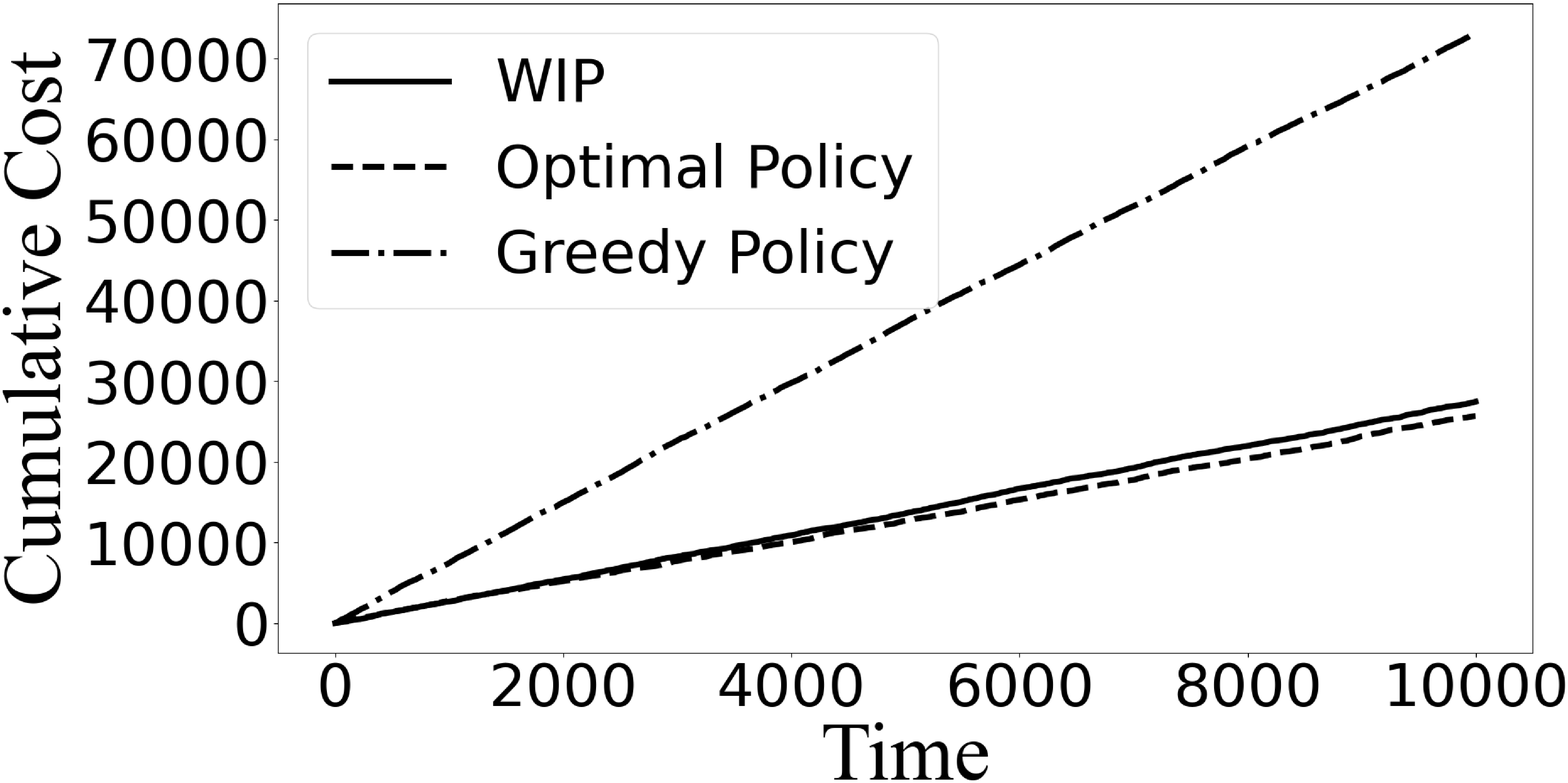}
        \caption{Performance of Whittle Index Policy for a problem conforming to Assumptions~\ref{assum:popularity-stoch-order} and~\ref{assum:transition-prob}}
        \label{fig: Optimal}
    \end{minipage}
\vspace{-0.2in}
\end{figure}

\paragraph*{Performance of the Whittle index policy  for a problem not conforming to Assumption~\ref{assum:transition-prob}} We manipulate the content missing costs of such a problem as suggested in Section~\ref{sec:nonindexable} to make it indexable. In other words, we choose $\hat{C}(1)$ and $\hat{C}(0)$ as in the Table \ref{table:modified-costs}. We obtain the Whittle indices of the modified problem and use them for the original problem. Then we compare the performances of the  Whittle index policy, the optimal policy and the greedy policy for the original problem. We see that the Whittle index policy still performs better compared to the greedy policy and is close to the optimal policy and is shown in Fig.~\ref{fig: Non-Index}

\paragraph{Indexability of the caching problem when $ \delta > \frac{1-\beta}{ \beta}$} The caching problem~\ref{eqn:objective} is seen to be indexable when $ \delta \leq \frac{1-\beta}{ \beta}$. Here we numerically check indexability when $ \delta > \frac{1-\beta}{ \beta}$. We use $p^0 = 0.1855 , p^1 = 0.2137, q^0 = 0.7719, q^1 =0.6280$ and $\beta = 0.95$, resulting in $\delta = 0.287                 
 8$ and $\frac{1-\beta}{ \beta} = 0.0526$.   We run value iteration for different values of $\lambda$ to find $Q_{\lambda}(a,r, a')  \quad \forall a',a,r$  and plot $Q_{\lambda}(a,r, 1) - Q_{\lambda}(a,r, 0) $ vs $\lambda$. For indexability $Q_{\lambda}(a,r, 1) - Q_{\lambda}(a,r, 0) $ should be non-decreasing in $\lambda$ as argued in the proof of Theorem~\ref{thm:indexability} . From Figure \ref{fig: Index-Action}, we can see that $Q_{\lambda}(a,r, 1) - Q_{\lambda}(a,r, 0) $ is indeed non-decreasing in $\lambda$. Hence the caching problem is indexable even when $ \delta > \frac{1-\beta}{ \beta}$.

\begin{figure}
\vspace{-0.2in}
    \centering
    \begin{minipage}[t]{0.45\textwidth}
        \centering
        \includegraphics[width=2.2in,height=1.5in]{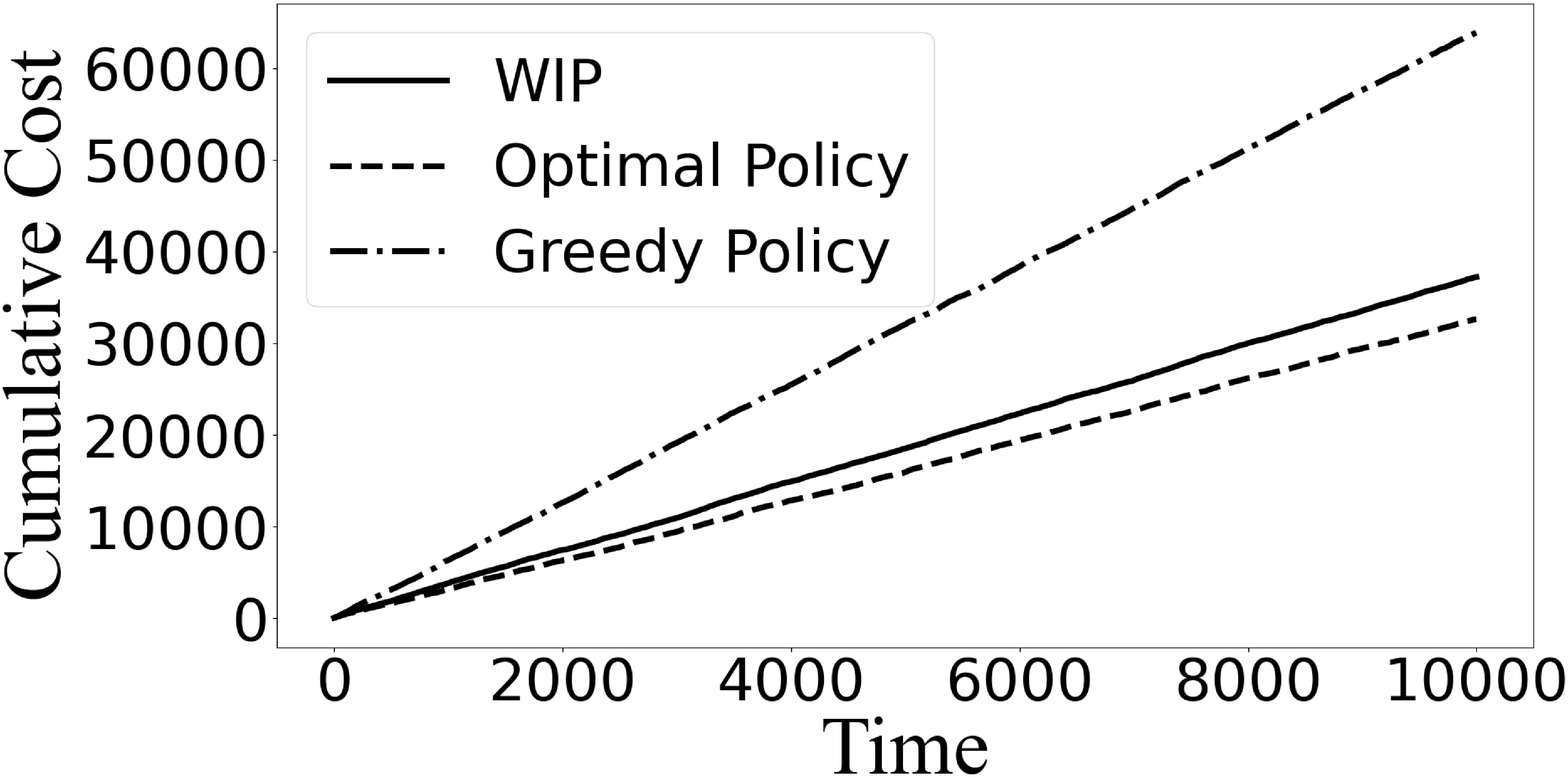}
        \caption{Performance of Whittle Index Policy for a problem not conforming to Assumption~\ref{assum:transition-prob}}
    \label{fig: Non-Index}
    \end{minipage}\hfill
    \begin{minipage}[t]{0.45\textwidth}
        \centering
        \includegraphics[width=2.2in,height=1.4in]{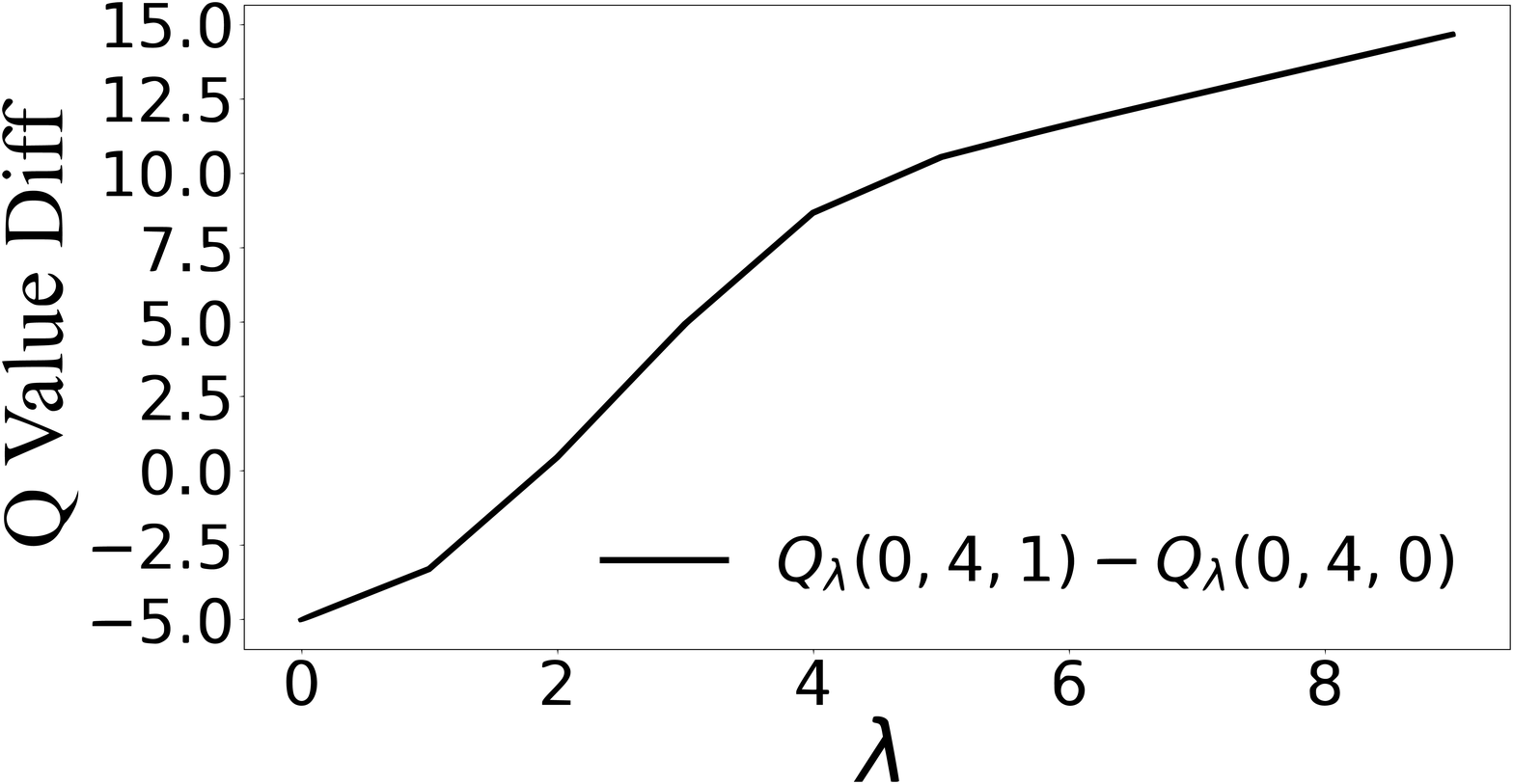}
         \caption{Indexability of caching problem when  $ \delta > \frac{1-\beta}{ \beta}$.}
    \label{fig: Index-Action}
    \end{minipage}
\vspace{-0.4in}
\end{figure}

\section{Conclusion}
\label{sec:conclusion}


We considered optimal caching of contents with varying popularity. We posed the problem as a discounted cost Markov decision problem and showed that it is a an instance of RMAB. We provided a condition under which its arms is indexable and also demonstarted the performance of the Whittle index policy. 

Our future work entails deriving performance bounds of the Whittle index policy. We plan to consider more general popularity dynamics. We also aspire to augment RMAB model with reinforcement learning to deal with the scenarios where the popularity dynamics might be unknown.
\remove{
\begin{enumerate}
    \item In our work, we have considered random walk-like state transitions, but more general transition structure can be considered in future.
    \item Since Whittle index policy is not an optimal policy in general, we can find a performance bound for this policy.
    \item In our case, we assume that underlying distribution is known, in reality this information is not available. Hence there is a need to develop an efficient online learning algorithm to make
optimal content caching decisions dynamic when these underlying distributions are unknown.
    \item In our work, we have considered a single cache problem, but in reality, there will be a network of caches. Hence we need to consider a more general caching network in future. 
\end{enumerate}
}

\bibliographystyle{IEEEtran}
\bibliography{Citation}

\appendix

\section{Proof of  Lemma \ref{lem:value-fun-concave}} \label{Proof:Lemma 1}

We prove the claims via induction on $n$. Recall that
$V_{\lambda}^0(a,r) = 0$ for all $r \in \mathbb{Z}_+$ and $a = 0,1$. So,
\begin{equation*}
V^1_{\lambda}(a,r) = \min_{a'\in\{0,1\}} Q^1_{\lambda}(a,r,a') =  \min\{T^0C(r),\lambda+d(1-a)\}
\end{equation*}
where $T^0C(0) = p^0 C(1) + (1-p^0) C(0)$ and $T^0C(r) = p^0 C(r+1) + q^0C(r-1) + (1-p^0-q^0) C(r)$ for all $r \geq 1$. It can be easily verified that $T^0C(r)$ is concave under Assumptions~\ref{assum:miss-hit-cost-concave} and~\ref{assum:transition-prob}. In particular, $T^0C(1)-T^0C(0) \geq T^0C(2)-T^0C(1)$ under Assumption~\ref{assum:transition-prob}, and $T^0C(r)-T^0C(r-1) \geq T^0C(r+1)-T^0C(r)$ for all $r \geq 2$ if $C(r)$ is concave as stated in Assumption~\ref{assum:miss-hit-cost-concave}. 
Therefore, $V^1_{\lambda}(a,r)$ being the minimum of a constant and a concave function, is also concave in $r$ for $a = 0,1$. Since $C(r)$ is non-decreasing, so are $T^0C(r)$ and $V^1_{\lambda}(a,r)$ for $a = 0,1$.  
Now suppose $V^n_{\lambda}(a,r)$ are concave and non-decreasing for $a = 0,1$ for some $n \geq 1$. Then, $T^{a'}V^n_{\lambda}(a',r)$ is concave and non-decreasing in $r$ for $a' = 0,1$, and consequently, \begin{align*}
Q^{n+1}_{\lambda}(a,r,a') &= c_{\lambda}(a,r,a')+\beta T^{a'}V^n_{\lambda}(a',r) \\
& = \lambda a'+ d a'(1-a) + (1-a')T^{a'}C(r) \\
&\ \ \ + \beta T^{a'}V^n_{\lambda}(a',r),
\end{align*}
being the sum of concave and non-decreasing functions is also concave and non-decreasing.
Finally,
\[V^{n+1}_{\lambda}(a,r) = \min_{a'\in\{0,1\}} Q^n_{\lambda}(a,r,a'), 
\]
Being minimum of two concave and non-decreasing functions is also concave and non-decreasing. This completes the induction step.

\section{Proof of  Lemma \ref{lem:value-fun-properties}} \label{Proof:Lemma 2}

We prove the claims via induction on $n$. 
Recall that $V_{\lambda}^0(a,r) = 0$ for all $r \in \mathbb{Z}_+$ and $a = 0,1$. So,
\begin{align*}
\lefteqn{Q^1_{\lambda}(a,r,0)-Q^1_{\lambda}(a,r,1)} \\
&= c_{\lambda}(a,r,0) - c_{\lambda}(a,r,1) \\
&=T^{0}C(r) - \lambda a - d(1-a)
\end{align*}
which is non-decreasing in $r$ because $C(r)$ is non-decreasing in $r$. For the induction step, we prove the following two assertions.

\noindent
$(a)$ {\it If $Q^n_{\lambda}(a,r,0)-Q^n_{\lambda}(a,r,1)$ for $a = 0,1$ are increasing in $r$, so is $V_{\lambda}^n(0,r) -  V_{\lambda}^n(1,r)$.} Observe that
\begin{align*}
\lefteqn{V_{\lambda}^n(0,r) -  V_{\lambda}^n(1,r)} \\
& = \min_{a'}Q^n_{\lambda}(0,r,a') - \min_{a'}Q^n_{\lambda}(1,r,a') \\
& = \min \left\{\max_{a'} \left(Q^n_{\lambda}(0,r,0) - Q^n_{\lambda}(1,r,a')\right) \right., \\
& \ \ \ \ \ \ \ \ \ \ \left. \max_{a'} \left(Q^n_{\lambda}(0,r,1) - Q^n_{\lambda}(1,r,a')\right) \right\}\\
& = \min \left\{\max\{0,Q^n_{\lambda}(1,r,0) - Q^n_{\lambda}(1,r,1)\}, \right. \\
& \ \ \ \ \ \ \ \ \ \ \left. \max\{ Q^n_{\lambda}(0,r,1) - Q^n_{\lambda}(1,r,0),d\}\right\}.
\end{align*}
Let us consider the following two cases separately.
\begin{enumerate}
\item $Q^n_{\lambda}(1,r,0) - Q^n_{\lambda}(1,r,1) \geq 0:$ In this case $Q^n_{\lambda}(0,r,1) - Q^n_{\lambda}(1,r,0) \leq d$. Hence 
\[
V_{\lambda}^n(0,r) -  V_{\lambda}^n(1,r) = \min\{Q^n_{\lambda}(1,r,0) - Q^n_{\lambda}(1,r,1),d\}
\]
Which is non-decreasing in $r$ from the induction hypothesis.
\item $Q^n_{\lambda}(1,r,0) - Q^n_{\lambda}(1,r,1) < 0:$ In this case $Q^n_{\lambda}(0,r,1) - Q^n_{\lambda}(1,r,0) > d$, and so, 
$V_{\lambda}^n(0,r) -  V_{\lambda}^n(1,r) = 0$,
which is trivially non-decreasing in $r$.
\end{enumerate}

\noindent
$(b)$ {\it If $V_{\lambda}^n(0,r) -  V_{\lambda}^n(1,r)$ is increasing in $r$, so are $Q^{n+1}_{\lambda}(a,r,0)-Q^{n+1}_{\lambda}(a,r,1)$ for $a = 0,1$.} Observe that
\begin{align*}
\lefteqn{Q^{n+1}_{\lambda}(a,r,0)-Q^{n+1}_{\lambda}(a,r,1)} \\
& = c_{\lambda}(a,r,0)+\beta T^{0}V^n{\lambda}(0,r) - c_{\lambda}(a,r,1) - \beta T^{1}V^n_{\lambda}(1,r) \\
&= T^{0}C(r) - \lambda - d(1-a)
 + \beta\left(T^{0}V^n_{\lambda}(0,r) - T^{1}V^n_{\lambda}(1,r) \right).
\end{align*}
As observed earlier $T^{0}C(r)$ in non-decreasing in $r$. We argue that $T^{0}V^n_{\lambda}(0,r) - T^{1}V^n_{\lambda}(1,r)$ is also non-decreasing in $r$ establishing the desired claim. Towards this notice that for any $r \in \mathbb{Z}_+$, 
\begin{align}
 \lefteqn{T^{1}V^n_{\lambda}(1,r+1) - T^{1}V^n_{\lambda}(1,r)} \nonumber  \ \ \ \\
 & \leq T^{0}V^n_{\lambda}(1,r+1) - T^{0}V^n_{\lambda}(1,r) \nonumber \\
 & \leq T^{0}V^n_{\lambda}(0,r+1) - T^{0}V^n_{\lambda}(0,r) \label{eqn:value-func-ineq}
 \end{align}
where the first inequality follows from Assumption~\ref{assum:popularity-stoch-order} and monotonicity of $V^n_{\lambda}(a,r+1)-V^n_{\lambda}(a,r)$  as shown in Lemma~\ref{lem:value-fun-concave}; these allow us to apply~\cite[Lemma~4.7.2]{puterman2014markov}. The second inequality follows from the induction hypothesis. Finally,~\eqref{eqn:value-func-ineq} can be written as
\[
T^{0}V^n_{\lambda}(0,r) - T^{1}V^n_{\lambda}(1,r) \leq T^{0}V^n_{\lambda}(1,r+1) - T^{1}V^n_{\lambda}(1,r+1),
\]
i.e., $T^{0}V^n_{\lambda}(0,r) - T^{1}V^n_{\lambda}(1,r)$ is non-decreasing in $r$ as desired.

Parts~$(a)$ and $(b)$ together complete the induction step, hence the proof.

\section{Proof of Theorem ~\ref{thm:threshold-policy}} \label{Proof: Theorem 1}

Recall from Lemma~\ref{lem:value-fun-properties} that $Q^n_{\lambda}(a,r,0)-Q^n_{\lambda}(a,r,1)$ are non-decreasing in $r$ for $a = 0,1$. Define
\begin{align*}
r^a(\lambda) = \min\{r: Q^n_{\lambda}(a,r,0) - Q^n_{\lambda}(a,r,1) > 0\} \label{eqn:policy-threshold}
\end{align*}
and
\begin{align}
u^\ast(a,r) = 
\begin{cases}
0 \text{ if } r \leq r^a(\lambda) \\
1 \text{ otherwise.}
\end{cases}
\end{align}
for $a = 0,1$. Clearly, 
\[
u^\ast(a,r) = \argmin_{a'\in \{0,1\}} Q_{\lambda}(a,r,a'). 
\]
and so, $u^\ast:\{0,1\} \times \mathbb{Z}_+ \to \{0,1\}$ is an optimal policy. Furthermore,
\[
Q^n_{\lambda}(0,r,0) - Q^n_{\lambda}(0,r,1) = Q^n_{\lambda}(1,r,0) - Q^n_{\lambda}(1,r,1) - d,
\]
implying that $r^0(\lambda) \geq r^1(\lambda)$.
\remove{
Hence, the desired claim can be equivalently expressed as follows. There exist $r^0, r^1$ such that, for $a = 0,1$, $Q_{\lambda}(a,r,0) \leq Q_{\lambda}(a,r,1)$ if $r \leq r^a$ and $Q_{\lambda}(a,r,0) > Q_{\lambda}(a,r,1)$ otherwise. Moreover, a sufficient condition for existence of such $r_0,r_1$ is $Q_{\lambda}(a,r,0)-Q_{\lambda}(a,r,1)$ be non-decreasing in $r$ for $a = 0,1$.
We will prove that $Q^n_{\lambda}(a,r,0)-Q^n_{\lambda}(a,r,1)$ are non-decreasing in $r$ for $a = 0,1$ and for all $n \geq 1$. Taking $n \to \infty$ establishes the desired sufficient condition. 
Observe that, for any $a \in \{0,1\}$ and $r \in \mathbb{Z}_+$,
\begin{align*}
\lefteqn{Q^{n+1}_{\lambda}(a,r,0)-Q^{n+1}_{\lambda}(a,r,1)} \\
&=  c_{\lambda}(a,r,0)+\beta T^{0}V^n{\lambda}(0,r) - c_{\lambda}(a,r,1) - \beta T^{1}V^n_{\lambda}(1,r) \\
&= T^{0}C(r) - \lambda - d(1-a)
 + \beta\left(T^{0}V^n_{\lambda}(0,r) - T^{1}V^n_{\lambda}(1,r) \right).
\end{align*}
$T^{0}C(r)$ is non-decreasing in $r$ because $C(r)$ is non-decreasing in $r$ 
and $T^{0}V^n_{\lambda}(0,r) - T^{1}V^n_{\lambda}(1,r)$ is non-decreasing in $r$ from Lemma~\ref{lem:value-fun-properties}. Hence,   $Q^{n+1}_{\lambda}(a,r,0)-Q^{n+1}_{\lambda}(a,r,1)$ is also  non-decreasing in $r$ which completes the proof.
}

\section{Proof of  Lemma \ref{lem:value-fun-lambda} } \label{Proof: Lemma 3}

We prove the claims via induction on $n$. 
Recall that $V_{\lambda}^0(a,r) = 0$ for all $r \in \mathbb{Z}_+$ and $a = 0,1$. So,
\begin{align*}
\lefteqn{Q^1_{\lambda}(a,r,1)-Q^1_{\lambda}(a,r,0)} \\
&= c_{\lambda}(a,r,1) - c_{\lambda}(a,r,0) \\
&=\lambda a + d(1-a) - T^{0}C(r) 
\end{align*}
which is non-decreasing in $\lambda$. For the induction step we prove the following three assertions.

\noindent
$(a)$ {\it If $Q^n_{\lambda}(a,r,1)-Q^n_{\lambda}(a,r,0)$ for $a = 0,1$ are non-decreasing in $\lambda$, so is $V_{\lambda}^n(1,r) -  V_{\lambda}^n(0,r)$.} Observe that
\begin{align*}
\lefteqn{V_{\lambda}^n(1,r) -  V_{\lambda}^n(0,r)} \\
& = \min_{a'}Q^n_{\lambda}(1,r,a') - \min_{a'}Q^n_{\lambda}(0,r,a') \\
& = \min \left\{\max_{a'} \left(Q^n_{\lambda}(1,r,1) - Q^n_{\lambda}(0,r,a')\right) \right., \\
& \ \ \ \ \ \ \ \ \ \ \left. \max_{a'} \left(Q^n_{\lambda}(1,r,0) - Q^n_{\lambda}(0,r,a')\right) \right\}\\
& = \min \left\{\max\{-d,Q^n_{\lambda}(1,r,1) - Q^n_{\lambda}(1,r,0)\}, \right. \\
& \ \ \ \ \ \ \ \ \ \ \left. \max\{ Q^n_{\lambda}(1,r,0) - Q^n_{\lambda}(0,r,1),0\}\right\}.
\end{align*}
Let us consider the following two cases separately.
\begin{enumerate}
\item $Q^n_{\lambda}(1,r,0) - Q^n_{\lambda}(0,r,1) \leq 0:$ In this case $Q^n_{\lambda}(1,r,1) - Q^n_{\lambda}(1,r,0) \geq -d$. Hence
\[
V_{\lambda}^n(1,r) -  V_{\lambda}^n(0,r) = \min\{Q^n_{\lambda}(1,r,1) - Q^n_{\lambda}(1,r,0),0\}
\]
which is non-decreasing in $\lambda$ from the induction hypothesis.
\item $Q^n_{\lambda}(1,r,0) - Q^n_{\lambda}(0,r,1) < 0:$ In this case $Q^n_{\lambda}(1,r,1) - Q^n_{\lambda}(1,r,0) < -d$, and so, $V_{\lambda}^n(0,r) -  V_{\lambda}^n(1,r) = -d$,
which is trivially non-decreasing in $\lambda$.
\end{enumerate}

\noindent
$(b)$ {\it If $V_{\lambda}^{n-1}(1,r) -  V_{\lambda}^{n-1}(0,r)$, $V_{\lambda}^{n-1}(a,r+1) -  V_{\lambda}^{n-1}(a,r)$ and $Q^n_{\lambda}(a,r,1)-Q^n_{\lambda}(a,r,0)$ for $a = 0,1$ are non-decreasing in $\lambda$, so is $V_{\lambda}^n(a,r+1) -  V_{\lambda}^n(a,r)$.} 
We consider the following four cases separately.
\begin{enumerate}
\item $Q^n_{\lambda}(a,r+1,0) < Q^n_{\lambda}(a,r+1,1)$ and $Q^n_{\lambda}(a,r,0) < Q^n_{\lambda}(a,r,1)$: In this case
\begin{align*}
\lefteqn{V_{\lambda}^n(a,r+1) -  V_{\lambda}^n(a,r)} \\
& = Q^n_{\lambda}(a,r+1,0) - Q^n_{\lambda}(a,r,0) \\
& = T^0 C(r+1) - T^0 C(r) \\
& \ \ \ +  \beta\left(T^0 V_\lambda^{n-1}(0,r+1) - T^0 V_\lambda^{n-1}(0,r) \right)
\end{align*}
which is non-decreasing in $\lambda$ since $V_\lambda^{n-1}(0,r+1) - V_\lambda^{n-1}(0,r)$ 
is non-decreasing in $\lambda$  from the induction hypothesis.
\item $Q^n_{\lambda}(a,r+1,1) < Q^n_{\lambda}(a,r+1,0)$ and $Q^n_{\lambda}(a,r,1) < Q^n_{\lambda}(a,r,0)$: In this case
\begin{align*}
\lefteqn{V_{\lambda}^n(a,r+1) -  V_{\lambda}^n(a,r)} \\
& = Q^n_{\lambda}(a,r+1,1) - Q^n_{\lambda}(a,r,1) \\
& = \beta\left(T^1 V_\lambda^{n-1}(1,r+1) - T^1 V_\lambda^{n-1}(1,r) \right)
\end{align*}
which is non-decreasing in $\lambda$ since $V_\lambda^{n-1}(1,r+1) - V_\lambda^{n-1}(1,r)$ 
is non-decreasing in $\lambda$  from the induction hypothesis.
\item $Q^n_{\lambda}(a,r+1,1) < Q^n_{\lambda}(a,r+1,0)$ and $Q^n_{\lambda}(a,r,0) < Q^n_{\lambda}(a,r,1)$: In this case
\begin{align*}
\lefteqn{V_{\lambda}^n(a,r+1) -  V_{\lambda}^n(a,r)} \\ 
& = Q^n_{\lambda}(a,r+1,1) - Q^n_{\lambda}(a,r,0) \\
& = Q^n_{\lambda}(a,r+1,1) - Q^n_{\lambda}(a,r,1) \\
& \ \ \ + Q^n_{\lambda}(a,r,1) - Q^n_{\lambda}(a,r,0) \\
& = \beta\left(T^1 V_\lambda^{n-1}(1,r+1) - T^1 V_\lambda^{n-1}(1,r) \right) \\
& \ \ \ +  Q^n_{\lambda}(a,r,1) - Q^n_{\lambda}(a,r,0)
\end{align*}
which is non-decreasing in $\lambda$ since both $V_\lambda^{n-1}(1,r+1) - V_\lambda^{n-1}(1,r)$ 
and $Q^n_{\lambda}(a,r,1) - Q^n_{\lambda}(a,r,0)$ for $a = 0,1$ are non-decreasing in $\lambda$  from the induction hypothesis.

\item $Q^n_{\lambda}(a,r+1,0) < Q^n_{\lambda}(a,r+1,1)$ and $Q^n_{\lambda}(a,r,1) < Q^n_{\lambda}(a,r,0)$: This case is not feasible as it violates the fact that $Q^n_{\lambda}(a,r,0) - Q^n_{\lambda}(a,r,1)$ is non-decreasing in $r$; see Lemma~\ref{lem:value-fun-properties}.
\end{enumerate}

Hence we see that in all the feasible cases $V_{\lambda}^n(a,r+1) -  V_{\lambda}^n(a,r)$ is non-decreasing in $\lambda$ as desired.

\noindent
$(c)$ {\it If $V_{\lambda}^n(1,r) -  V_{\lambda}^n(0,r)$ and $V_{\lambda}^n(a,r+1) -  V_{\lambda}^n(a,r)$ for $a = 0,1$ are non-decreasing in $\lambda$, so are $Q^{n+1}_{\lambda}(a,r,1)-Q^{n+1}_{\lambda}(a,r,0)$ for $a = 0,1$.} Observe that
\begin{align*}
\lefteqn{Q^{n+1}_{\lambda}(a,r,1)-Q^{n+1}_{\lambda}(a,r,0)} \\
& = c_{\lambda}(a,r,1)+\beta T^{1}V^n{\lambda}(1,r) - c_{\lambda}(a,r,0) - \beta T^{0}V^n_{\lambda}(0,r) \\
&= \lambda + d(1-a) -T^{0}C(r) 
+ \beta\left(T^{1}V^n_{\lambda}(1,r) - T^{0}V^n_{\lambda}(0,r) \right).\\
&= \lambda + d(1-a) -T^{0}C(r) 
+ \beta\left(T^{1}V^n_{\lambda}(1,r) - T^{1}V^n_{\lambda}(0,r) \right) \\
& \ \  \ + \beta\left(T^{1}V^n_{\lambda}(0,r) -  T^{0}V^n_{\lambda}(0,r) \right))
\end{align*}
We argue that both $T^{1}V^n_{\lambda}(1,r) - T^{1}V^n_{\lambda}(0,r)$ and $T^{1}V^n_{\lambda}(0,r) -  T^{0}V^n_{\lambda}(0,r)$ are non-decreasing in $\lambda$ establishing the desired claim. 
\begin{enumerate}
\item To show that $T^{1}V^n_{\lambda}(1,r) - T^{1}V^n_{\lambda}(0,r)$ is non-decreasing in $\lambda$, it is enough to show that $V^n_{\lambda}(1,r) - V^n_{\lambda}(0,r)$ is non-decreasing in $\lambda$ for all $r \in \mathbb{Z}_+$. The latter assertion holds from the induction hypothesis.

\item $T^{1}V^n_{\lambda}(0,r) -  T^{0}V^n_{\lambda}(0,r)$ being non-decreasing in $\lambda$ is equivalent to 
\[
T^{1}V^n_{\lambda'}(0,r) - T^{1}V^n_{\lambda}(0,r) \geq T^{0}V^n_{\lambda'}(0,r) - T^{0}V^n_{\lambda}(0,r)
\]
for all $\lambda' \geq \lambda$. The letter assertion holds if $V^n_{\lambda'}(0,r) - V^n_{\lambda}(0,r)$ is non-decreasing in $r$ which in turn is equivalent to $V^n_{\lambda}(0,r+1) - V^n_{\lambda}(0,r)$ being non-decreasing in $\lambda$.   The last assertion holds from the induction hypothesis.
\end{enumerate}

Parts~$(a)$, $(b)$ and $(c)$ together complete the induction step and hence the proof.

\section{Proof of Theorem ~\ref{thm:indexability}} \label{Proof: Theorem 2}

From Lemma~\ref{lem:value-fun-lambda}, $Q^n_{\lambda}(a,r,1)-Q^n_{\lambda}(a,r,0)$ are non-decreasing in $\lambda$ for $a = 0,1$ and for all $n \geq 1$. Hence, taking $n \to \infty$, $Q_{\lambda}(a,r,1)-Q_{\lambda}(a,r,0)$ is also non-decreasing in $\lambda$. Following similar arguments as in the proof of Theroem~\ref{thm:threshold-policy}, it implies that each arm of the RMAB formulation of the content caching problem is indexable as required.

\section{Proof of Lemma \ref{lem:modified-indexability}} \label{Proof: Lemma 4}
\textbf{Case 1:} $p^0 + 2q^0 < 1$

In this case, we will show that $\hat{C}:\mathbb{Z}_+ \to \mathbb{R}_+$ satisfies all  three properties.

From the Table \ref{table:modified-costs}, we can see that 
\begin{equation}
     \hat{C}(1) = C(1)  \label{New C1}
\end{equation}
\begin{align}
     \hat{C}(0) &= \min\Bigg\{ C(0) ,C(1)  - \nonumber \\
    & \Bigg[\frac{p^0( ( C(2) - C(1))-(C(3) - C(2)) )}{  p^0 + 2 q^0 -1}  \nonumber \\
      &- \frac{( 1-p^0-q^0 )( C(2) - C(1))}{ p^0 + 2 q^0 -1} \Bigg] \Bigg\}
\end{align}

\textbf{Property 1:} $\hat{C}:\mathbb{Z}_+ \to \mathbb{R}_+$ is non-decreasing

To show that   above choices of $\hat{C}(0)$ and $\hat{C}(1) $ satisfies property 1, it is enough to show that $\hat{C}(0) \leq \hat{C}(1) \leq C(2)$ because $\hat{C}(r) = C(r), r \geq 2$ is non-decreasing due to assumption \ref{assum:miss-hit-cost-concave}.

From the choice of $\hat{C}(0)$,  we can see that $\hat{C}(0) \leq C(0)$. $C(0) \leq C(1) \leq C(2) \leq ...$ because of assumption \ref{assum:miss-hit-cost-concave}. Also $\hat{C}(1) = C(1)  $. Hence
$$ \hat{C}(0) \leq C(0) \leq C(1) =  \hat{C}(1) \leq C(2) \leq ... $$ Therefore $\hat{C}:\mathbb{Z}_+ \to \mathbb{R}_+$ is non-decreasing

\textbf{Property 2:} $\hat{C}:\mathbb{Z}_+ \to \mathbb{R}_+$ is concave, i.e.  $\hat{C}(1) - \hat{C}(0) \geq  C(2) - \hat{C}(1) \geq C(3) - C(2)  \geq ......$

To show that   above choices of $\hat{C}(0)$ and $\hat{C}(1) $ satisfies property 2, it is enough to show that $\hat{C}(1) - \hat{C}(0) \geq C(2) - \hat{C}(1) \geq C(3) - C(2) $ because $\hat{C}(r) = C(r), r \geq 2$ is concave and has decreasing difference due to assumption \ref{assum:miss-hit-cost-concave}.

Since $\hat{C}(0) \leq C(0)$, we have
\begin{align}
    \hat{C}(1) - \hat{C}(0) \geq C(1) - C(0) \label{Con 1}
\end{align}
From assumption \ref{assum:miss-hit-cost-concave}, we have
\begin{align}
    C(1) - C(0) \geq C(2) - C(1) \geq C(3) - C(2) \geq ... \label{Con2}
\end{align}
From equation \ref{New C1}, \ref{Con 1}, \ref{Con2} , we have  the following $$ \hat{C}(1) - \hat{C}(0) \geq  C(2) - \hat{C}(1) \geq C(3) - C(2)  \geq ......$$
Therefore  $\hat{C}:\mathbb{Z}_+ \to \mathbb{R}_+$ is concave.

\textbf{Property 3:}$p^0( C(3) - C(2) ) - ( 2p^0 + q^0 -1 )( C(2) - \hat{C}(1)) +   \\
( p^0 + 2 q^0 -1)( \hat{C}(1) -\hat{C}(0) ) \leq 0  $

Since $p^0 + 2q^0 < 1$, assumption \ref{assum:transition-prob} becomes the following

\begin{align} \label{New C0 2}
    \hat{C}(0)   &\leq \hat{C}(1)  - \Bigg[\frac{p^0( ( C(2) -\hat{C}(1))-(C(3) - C(2)) )}{  p^0 + 2 q^0 -1}  \nonumber  \\
      &- \frac{( 1-p^0-q^0 )( C(2) - \hat{C}(1))}{ p^0 + 2 q^0 -1} \Bigg]   \nonumber \\
 \end{align}

Now, From the choice of $\hat{C}(0)$,  we can see that 

\begin{align} \label{Assum 3 Check }
    \hat{C}(0) \leq  \hat{C}(1) - \Bigg[\frac{p^0( ( C(2) -\hat{C}(1))-(C(3) - C(2)) )}{  p^0 + 2 q^0 -1}  - \nonumber \\
 \frac{( 1-p^0-q^0 )( C(2) - \hat{C}(1)}{ p^0 + 2 q^0 -1} \Bigg]
\end{align}

Since equation \ref{New C0 2} and \ref{Assum 3 Check } are the same, assumption \ref{assum:transition-prob} is satisfied. Hence property 3 is also satisfied. 

\textbf{Case 2:}  $p^0 + 2q^0 > 1,  2p^0 + q^0 < 1 $

In this case, we will show that property 3 is violated for those which satisfies property 1 and 2  and there do not exist 
$\hat{C}(0)$ and $\hat{C}(1)$ such that $\hat{C}:\mathbb{Z}_+ \to \mathbb{R}_+$ satisfies all these three properties.

Since $p^0 + 2q^0 > 1,  2p^0 + q^0 < 1 $
\begin{align*}
    p^0( C(3) - C(2) ) - ( 2p^0 + q^0 -1 )( C(2) - \hat{C}(1)) + \\  
( p^0 + 2 q^0 -1)(  \hat{C}(1) -\hat{C}(0) ) > 0
\end{align*}
Therefore, any choice of $C(0)$ and $C(1)$ that satisfies properties 1 and 2, will not  satisfy property 3. Therefore there do not exist 
$\hat{C}(0)$ and $\hat{C}(1)$ such that $\hat{C}:\mathbb{Z}_+ \to \mathbb{R}_+$ satisfies all these three properties.

\textbf{Case 3a:} $p^0 + 2q^0 > 1, 2p^0 + q^0 \geq 1, q^0 > p^0$

In this case, we will show that property 3 is violated for those which satisfies property 1 and 2  and there do not exist 
$\hat{C}(0)$ and $\hat{C}(1)$ such that $\hat{C}:\mathbb{Z}_+ \to \mathbb{R}_+$ satisfies all these three properties.

 Since  $ q^0 > p^0 $, we have the following
\begin{align*}
    2p^0 + q^0 - 1 &< p^0 + 2q^0 - 1 \\ 
    (C(2) - \hat{C}(1))( 2p^0 + q^0 - 1) &< (\hat{C}(1) - \hat{C}(0))(p^0 + 2q^0 - 1)
\end{align*}
The second inequality is due to $$ \hat{C}(1) - \hat{C}(0) \geq  C(2) -\hat{C}(1)
$$ because we assumed that properties 1 and 2 are satisfied. 
Therefore,
\begin{equation*}
     (\hat{C}(1) - \hat{C}(0))(p^0 + 2q^0 - 1)- 
    (C(2) -\hat{C}(1))( 2p^0 + q^0 - 1)  >0\\ 
\end{equation*}
\begin{align*}
    (\hat{C}(1) - \hat{C}(0))(p^0 + 2q^0 - 1) - \\
    (C(2) - \hat{C}(1))( 2p^0 + q^0 - 1) + 
    p^0( C(3) - C(2) ) >0
\end{align*}
Because $p^0 > 0$ and $C(3) \geq C(2)$ . Therefore, any choice of $C(0)$ and $C(1)$ that satisfies properties 1 and 2, will not  satisfy property 3. Therefore there do not exist 
$\hat{C}(0)$ and $\hat{C}(1)$ such that $\hat{C}:\mathbb{Z}_+ \to \mathbb{R}_+$ satisfies all these three properties.

\textbf{Case 3b:}  $p^0 + 2q^0 > 1, 2p^0 + q^0 > 1, q^0 < p^0$

In this case, we will show that $\hat{C}:\mathbb{Z}_+ \to \mathbb{R}_+$ satisfies all  three properties.

From the Table \ref{table:modified-costs}, we can see that 
\begin{align} 
    \hat{C}(0) &= 2 \hat{C}(1) - C(2) \label{Case 3b1} \\
    \hat{C}(1) &= C(2) - \frac{p^0}{p^0 - q^0} (  C(3) - C(2) ) \label{Case 3b2}
\end{align}
\textbf{Property 1:} $\hat{C}:\mathbb{Z}_+ \to \mathbb{R}_+$ is non-decreasing

To show that   above choices of $\hat{C}(0)$ and $\hat{C}(1) $ satisfies property 1, it is enough to show that $\hat{C}(0) \leq \hat{C}(1) \leq C(2)$ because $\hat{C}(r) = C(r), r \geq 2$ is non-decreasing due to assumption \ref{assum:miss-hit-cost-concave}.

We will show first $\hat{C}(1) - \hat{C}(0) \geq 0$
\begin{align*}
    \hat{C}(1) - \hat{C}(0) &= \{C(2) - \frac{p^0}{p^0 - q^0} (  C(3) - C(2) ) \} \\
    &- \{2 \hat{C}(1) - C(2) \}\\
    &= 2 \{ C(2) -   \hat{C}(1) \} - \frac{p^0}{p^0 - q^0} (  C(3) - C(2) )
\end{align*}

Now substitute for $\hat{C}(1)$ from \ref{Case 3b2} in the above expression, we get the following
\begin{align} \label{Monotone 1}
     \hat{C}(1) - \hat{C}(0) &= \frac{p^0}{p^0 - q^0} (  C(3) - C(2) ) \geq 0 
\end{align}
Because of assumption \ref{assum:miss-hit-cost-concave} i.e.,$C(3) \geq C(2)$ and $q^0 < p^0$

Now we will show $ C(2) -\hat{C}(1) \geq 0 $
\begin{align} \label{Monotone 2}
    C(2) -\hat{C}(1) &= \frac{p^0}{p^0 - q^0} (  C(3) - C(2) ) \geq 0 
\end{align}

From \ref{Monotone 1} and \ref{Monotone 2}, $\hat{C}(0) \leq \hat{C}(1) \leq C(2)$. Now  $\hat{C}(r) = C(r), r \geq 2$ is non-decreasing because of assumption \ref{assum:miss-hit-cost-concave}.  Therefore $\hat{C}:\mathbb{Z}_+ \to \mathbb{R}_+$ is non-decreasing.

\textbf{Property 2:} $\hat{C}:\mathbb{Z}_+ \to \mathbb{R}_+$ is concave, i.e.  $\hat{C}(1) - \hat{C}(0) \geq  C(2) - \hat{C}(1) \geq C(3) - C(2)  \geq ......$

To show that   above choices of $\hat{C}(0)$ and $\hat{C}(1) $ satisfies property 2, it is enough to show that $\hat{C}(1) - \hat{C}(0) \geq C(2) - \hat{C}(1) \geq C(3) - C(2) $ because $\hat{C}(r) = C(r), r \geq 2$ is concave and has decreasing difference due to assumption \ref{assum:miss-hit-cost-concave}.
From \ref{Case 3b1}, we have
\begin{align} \label{Concave 1}
    \hat{C}(1) - \hat{C}(0) &= C(2) - \hat{C}(1)
\end{align}

From \ref{Monotone 2}, we have 
\begin{align} \label{Concave 2}
    C(2) -\hat{C}(1) &= \frac{p^0}{p^0 - q^0} (  C(3) - C(2) )  \nonumber\\
    &\geq C(3) - C(2)  
\end{align}
because $ p^0 > q^0$ and $q^0 > 0$ .  

From \ref{Concave 1} and \ref{Concave 2}, we have the following
\begin{align*}
    \hat{C}(1) - \hat{C}(0) = C(2) - \hat{C}(1) \geq C(3) - C(2) 
\end{align*}
Now $\hat{C}(r) = C(r), r \geq 2$ is concave and has decreasing difference due to assumption \ref{assum:miss-hit-cost-concave}. Therefore property 2 is satisfied.

\textbf{Property 3:}$p^0( C(3) - C(2) ) - ( 2p^0 + q^0 -1 )( C(2) - \hat{C}(1)) +   \\
( p^0 + 2 q^0 -1)( \hat{C}(1) -\hat{C}(0) ) \leq 0  $

In this case, we show that the above choices of $\hat{C}(0)$ and $\hat{C}(1) $ satisfy property 3 with equality. 

From \ref{Concave 1} and \ref{Case 3b2}, we have 

\begin{align*}
    p^0( C(3) - C(2) ) - ( 2p^0 + q^0 -1 )( C(2) - \hat{C}(1)) + \\
( p^0 + 2 q^0 -1)( \hat{C}(1) -\hat{C}(0) ) \\
=\\
 p^0( C(3) - C(2) -  ( 2p^0 + q^0 -1 ) ( \frac{p^0}{p^0 - q^0} (  C(3) - C(2) )   ) + \\
 (p^0 + 2 q^0 -1  )( \frac{p^0}{p^0 - q^0} (  C(3) - C(2) )   ) = 0\\
\end{align*}

Therefore property 3 is satisfied.

\section{Proof of Indexability when $p^1 \approx p^0$ and $q^1 \approx q^0$}
\begin{lemma} \label{Val Function Lambda}
    $V_{\lambda}^n(a,r)$ is increasing in $\lambda \quad \forall a , r$ 
\end{lemma}
\begin{proof}
    It can be easily shown using induction and using the fact that minimum of two increasing functions is also increasing.
\end{proof}

\begin{lemma}
    $V_{\lambda_2}^n(1,r) -  V_{\lambda_1}^n(1,r) \leq \frac{\lambda_2 - \lambda_1}{1 - \beta} \quad \forall n , r $ and $\lambda_2 \geq \lambda_1$ 
\end{lemma}
\begin{proof}
     We prove the claims via induction on $n$. 
Recall that $V_{\lambda}^0(1,r) = 0$ for all $r \in \mathbb{Z}_+$. So,
\begin{align}
    V_{\lambda_2}^1 ( 1, r) - V_{\lambda_1}^1 (1,r) &= \min \{ \lambda_2 , T^0 C(r)\} - \min \{ \lambda_1 , T^0 C(r)\} \nonumber \\ 
    &\ineqa \lambda_2 - \lambda_1  \label{E1}
\end{align}

Above inequality (a) holds because $V_{\lambda_2}^1 ( 1, r) - V_{\lambda_1}^1 (1,r)$ takes values in the set  $\{0, \lambda_2 - \lambda_1, T^0 C(r) - \lambda_1 \} $. All of these values are less than $ \lambda_2 - \lambda_1$.
Similarly,  we can show the following also
\begin{align} \label{E10}
    V_{\lambda_2}^1 ( 0, r) - V_{\lambda_1}^1 (0,r) \leq \lambda_2 - \lambda_1
\end{align}

\begin{align*}
    V_{\lambda_2}^2 ( 1, r) - V_{\lambda_1}^2 (1,r) = \min \{ \lambda_2 + \beta T^1 V_{\lambda_2}^1 (1, r), T^0 C(r) +
     \beta T^0 V_{\lambda_2}^1 ( 0, r) \} - \\
    \min \{ \lambda_1 + \beta T^1 V_{\lambda_1}^1 (1, r) , T^0 C(r) +
     \beta T^0 V_{\lambda_1}^1 ( 0, r)  \} 
\end{align*}

From Lemma~\ref{Val Function Lambda} , we know 
\begin{align*}
    \lambda_2 + \beta T^1 V_{\lambda_2}^1 (1, r) &\geq \lambda_1 + \beta T^1 V_{\lambda_1}^1 (1, r) \\
    T^0 C(r) + \beta T^0 V_{\lambda_2}^1 ( 0, r) &\geq T^0 C(r) +
     \beta T^0 V_{\lambda_1}^1 ( 0, r) 
\end{align*}

\begin{enumerate}
        \item \textbf{Case 1:}  $ \lambda_2 + \beta T^1 V_{\lambda_2}^1 (1, r) \leq T^0 C(r) + \beta T^0 V_{\lambda_2}^1 ( 0, r) $ and $T^0 C(r) +
     \beta T^0 V_{\lambda_1}^1 ( 0, r)  \geq \lambda_1 + \beta T^1 V_{\lambda_1}^1 (1, r)$

     \begin{align}
          V_{\lambda_2}^2 ( 1, r) - V_{\lambda_1}^2 (1,r) &= \lambda_2 + \beta T^1 V_{\lambda_2}^1 (1, r) - \lambda_1 - \beta T^1 V_{\lambda_1}^1 (1, r) \nonumber \\
          &= ( \lambda_2 - \lambda_1) +\beta T^1 (  V_{\lambda_2}^1 (1, r) - V_{\lambda_1}^1 (1, r)) \nonumber\\ 
          &\ineqa ( \lambda_2 - \lambda_1) + \beta T^1( \lambda_2 - \lambda_1) \nonumber\\
          &= ( \lambda_2 - \lambda_1) ( 1 + \beta) \nonumber
     \end{align}
     where (a) due to Equation \ref{E1} 
     \item  \textbf{Case 2:} $ \lambda_2 + \beta T^1 V_{\lambda_2}^1 (1, r) \leq T^0 C(r) + \beta T^0 V_{\lambda_2}^1 ( 0, r) $ and $T^0 C(r) +
     \beta T^0 V_{\lambda_1}^1 ( 0, r)  \leq \lambda_1 + \beta T^1 V_{\lambda_1}^1 (1, r)$
     \begin{align*}
         V_{\lambda_2}^2 ( 1, r) - V_{\lambda_1}^2 (1,r) &= \lambda_2 + \beta T^1 V_{\lambda_2}^1 (1, r) - T^0 C(r) -\beta T^0 V_{\lambda_1}^1 ( 0, r) \\
         &\ineqa T^0 C(r) + \beta T^0 V_{\lambda_2}^1 ( 0, r)  - T^0 C(r) - \beta T^0 V_{\lambda}^1 ( 0, r-1) \\
         &=  \beta T^0 [ V_{\lambda_2}^1 ( 0, r)  - V_{\lambda_1}^1 ( 0, r) ]\\
         &\ineqb \beta T^0 (  \lambda_2 - \lambda_1) \\
         &= \beta(\lambda_2 - \lambda_1 )
         \end{align*}
     where (a) is due to the inequality in \textbf{Case 2} and (b) due to equation \ref{E10}. 
\end{enumerate}
Similarly, if we go through all the possible cases, we can show that
\begin{align} \label{E2}
     V_{\lambda_2}^2 ( 1, r) - V_{\lambda_1}^2 (1,r) \leq ( \lambda_2 - \lambda_1) ( 1 + \beta) 
\end{align}
For any $n$, we can show the following
\begin{align}
    V_{\lambda_2}^n ( 1, r) - V_{\lambda_1}^n (1,r) &\leq ( \lambda_2 - \lambda_1) ( 1 + \beta + \beta^2 + ...+ \beta^{n-1})  \nonumber \\
    &\leq ( \lambda_2 - \lambda_1) \sum_{i =0}^{\infty}\beta^i \nonumber \\
    &= \frac{\lambda_2 - \lambda_1}{1- \beta}
\end{align}
\end{proof}

\begin{lemma} For all $n \geq 1$,
\label{lem:value-fun-lambda1}
 \begin{enumerate}
\item $Q^n_{\lambda}(a,r,1)-Q^n_{\lambda}(a,r,0)$ are non-decreasing in $\lambda$ for $a = 0,1$.
\item $V_{\lambda}^n(1,r) -  V_{\lambda}^n(0,r)$ are non-decreasing in $\lambda$. 
$\lambda$ for $a = 0,1$.
\end{enumerate}
\end{lemma}

\begin{proof}
    We prove the claims via induction on $n$. 
Recall that $V_{\lambda}^0(a,r) = 0$ for all $r \in \mathbb{Z}_+$ and $a = 0,1$. So,
\begin{align*}
\lefteqn{Q^1_{\lambda}(a,r,1)-Q^1_{\lambda}(a,r,0)} \\
&= c_{\lambda}(a,r,1) - c_{\lambda}(a,r,0) \\
&=\lambda a + d(1-a) - T^{0}C(r) 
\end{align*}
which is non-decreasing in $\lambda$. For the induction step we prove the following three assertions.

\noindent
$(a)$ {\it If $Q^n_{\lambda}(a,r,1)-Q^n_{\lambda}(a,r,0)$ for $a = 0,1$ are non-decreasing in $\lambda$, so is $V_{\lambda}^n(1,r) -  V_{\lambda}^n(0,r)$.} Observe that
\begin{align*}
\lefteqn{V_{\lambda}^n(1,r) -  V_{\lambda}^n(0,r)} \\
& = \min_{a'}Q^n_{\lambda}(1,r,a') - \min_{a'}Q^n_{\lambda}(0,r,a') \\
& = \min \left\{\max_{a'} \left(Q^n_{\lambda}(1,r,1) - Q^n_{\lambda}(0,r,a')\right), \right. \\
& \ \ \ \ \ \ \ \ \ \ \left. \max_{a'} \left(Q^n_{\lambda}(1,r,0) - Q^n_{\lambda}(0,r,a')\right) \right\}\\
& = \min \left\{\max\{-d,Q^n_{\lambda}(1,r,1) - Q^n_{\lambda}(1,r,0)\}, \right. \\
& \ \ \ \ \ \ \ \ \ \ \left. \max\{ Q^n_{\lambda}(1,r,0) - Q^n_{\lambda}(0,r,1),0\}\right\}.
\end{align*}
Let us consider the following two cases separately.
\begin{enumerate}
\item $Q^n_{\lambda}(1,r,0) - Q^n_{\lambda}(0,r,1) \leq 0:$ In this case $Q^n_{\lambda}(1,r,1) - Q^n_{\lambda}(1,r,0) \geq -d$. Hence
\[
V_{\lambda}^n(1,r) -  V_{\lambda}^n(0,r) = \min\{Q^n_{\lambda}(1,r,1) - Q^n_{\lambda}(1,r,0),0\}
\]
which is non-decreasing in $\lambda$ from the induction hypothesis.
\item $Q^n_{\lambda}(1,r,0) - Q^n_{\lambda}(0,r,1) > 0:$ In this case $Q^n_{\lambda}(1,r,1) - Q^n_{\lambda}(1,r,0) < -d$, and so, $V_{\lambda}^n(0,r) -  V_{\lambda}^n(1,r) = -d$,
which is trivially non-decreasing in $\lambda$.
\end{enumerate}
\noindent
$(b)$ {\it If $V_{\lambda}^n(1,r) -  V_{\lambda}^n(0,r)$ is non-decreasing in $\lambda$ and  $\delta \leq \frac{1-\beta}{\beta}$, so are $Q^{n+1}_{\lambda}(a,r,1)-Q^{n+1}_{\lambda}(a,r,0)$ for $a = 0,1$.} Let $\lambda_2 \geq \lambda_1$.  Observe that
\begin{align*}
\lefteqn{\{Q^{n+1}_{\lambda_2}(a,r,1)-Q^{n+1}_{\lambda_2}(a,r,0)\}- \{Q^{n+1}_{\lambda_1}(a,r,1)-Q^{n+1}_{\lambda_1}(a,r,0) \}} \\
&= \{\lambda_2 + d(1-a) -T^{0}C(r) 
+ \beta\left(T^{1}V^n_{\lambda_2}(1,r) - T^{0}V^n_{\lambda_2}(0,r) \right)\}\\
&- \{\lambda_1 + d(1-a) -T^{0}C(r) 
+ \beta\left(T^{1}V^n_{\lambda_1}(1,r) - T^{0}V^n_{\lambda_1}(0,r) \right)\}\\
&= (\lambda_2 - \lambda_1) + \beta T^{1}\left(V^n_{\lambda_2}(1,r) - V^n_{\lambda_1}(1,r) \right) -  \beta T^{0}\left(V^n_{\lambda_2}(1,r) - V^n_{\lambda_1}(1,r) \right)\\
&+ \beta T^{0}\left(V^n_{\lambda_2}(1,r) - V^n_{\lambda_1}(1,r) \right) - \beta T^{0}\left(V^n_{\lambda_2}(0,r) - V^n_{\lambda_1}(0,r) \right) \\
&= (\lambda_2 - \lambda_1) + \beta ( p^1 - p^0) (V^n_{\lambda_2}(1,r+1) - V^n_{\lambda_1}(1,r+1)   ) \\
&+ \beta ( q^1 - q^0 ) (V^n_{\lambda_2}(1,r-1) - V^n_{\lambda_1}(1,r-1) ) \\
&+ \beta ( p^0 - p^1 )  (V^n_{\lambda_2}(1,r) - V^n_{\lambda_1}(1,r)   ) + \beta ( q^0 - q^1 )  (V^n_{\lambda_2}(1,r) - V^n_{\lambda_1}(1,r)   )\\
&+ \beta T^0 \bigg\{  (V^n_{\lambda_2}(1,r) - V^n_{\lambda_2}(0,r)   ) -   (V^n_{\lambda_1}(1,r) - V^n_{\lambda_1}(0,r)   ) \bigg\} \\
&\ineqaa (\lambda_2 - \lambda_1) + \beta ( p^1 - p^0) (V^n_{\lambda_2}(1,r+1) - V^n_{\lambda_1}(1,r+1)   ) \\
&+ \beta ( q^1 - q^0 ) (V^n_{\lambda_2}(1,r-1) - V^n_{\lambda_1}(1,r-1) ) \\
&+ \beta ( p^0 - p^1 )  (V^n_{\lambda_2}(1,r) - V^n_{\lambda_1}(1,r)   ) + \beta ( q^0 - q^1 )  (V^n_{\lambda_2}(1,r) - V^n_{\lambda_1}(1,r)   )\\
&\ineqbb (\lambda_2 - \lambda_1) - \beta \delta \max_r \{  V^n_{\lambda_2}(1,r) - V^n_{\lambda_1}(1,r) \} \\
\end{align*}

where (a) and (b) holds after removing positive terms

We know that 
\begin{align} \label{upper bound}
    V_{\lambda_2}^n(1,r) -  V_{\lambda_1}^n(1,r) \leq \frac{\lambda_2 - \lambda_1}{1 - \beta}
\end{align}

\begin{align}
    &\{Q^{n+1}_{\lambda_2}(a,r,1)-Q^{n+1}_{\lambda_2}(a,r,0)\} -  
     \{Q^{n+1}_{\lambda_1}(a,r,1)-Q^{n+1}_{\lambda_1}(a,r,0) \} \nonumber \\
    &\ineqaa (\lambda_2 - \lambda_1) - \beta \delta \bigg( \frac{(\lambda_2 - \lambda_1)}{1 - \beta}\bigg)  \nonumber \\
    &=  (\lambda_2 - \lambda_1) \bigg( 1 - \frac{\beta \delta}{1 - \beta}\bigg)
\end{align}

where (a) holds due to equation \ref{upper bound}
For the above inequality to be greater than zero , we need 
\begin{align} \label{condition: delta}
    \delta \leq \frac{1 - \beta}{\beta} 
\end{align}
Parts~$(a)$ and $(b)$  together complete the induction step and hence the proof.
\end{proof}

From Lemma~\ref{lem:value-fun-lambda1}, $Q^n_{\lambda}(a,r,1)-Q^n_{\lambda}(a,r,0)$ are non-decreasing in $\lambda$ for $a = 0,1$ and for all $n \geq 1$. Hence, taking $n \to \infty$, $Q_{\lambda}(a,r,1)-Q_{\lambda}(a,r,0)$ is also non-decreasing in $\lambda$. Following similar arguments as in the proof of Theroem~\ref{thm:threshold-policy}, it implies that each arm of the RMAB formulation of the content caching problem is indexable as required.

\end{document}